%% file: main.tex
\documentclass[journal,12pt,draftclsnofoot,onecolumn]{IEEEtran}

\IEEEoverridecommandlockouts

\usepackage{setspace} 
 \usepackage{cite}
\ifCLASSINFOpdf
\fi
\usepackage{amsmath,amsfonts,amsthm} 
\usepackage{amssymb}
\usepackage{bbm}
\usepackage{float}
\usepackage{graphicx}
\usepackage{algorithm}
\usepackage{algpseudocode}
\usepackage{subcaption}
\usepackage{siunitx}
\usepackage{tikz}

\usetikzlibrary{arrows,automata}
\usepackage{multicol}
\usepackage{comment}
\usepackage{pgfplots}
\usepackage{epstopdf}
\usepackage{multirow}
\usepackage[font=scriptsize,skip=5pt]{caption}
\usepackage{soul}
\usepackage{todonotes}
\usepackage{color, colortbl}
\usepackage{array}
\usetikzlibrary{patterns}
\usepackage{balance}

\usepackage{enumitem}
\usepackage{footmisc}
\usepackage{pgfplots}
\usepackage{caption}
\usepackage{subcaption}
\usepackage{wrapfig}
\allowdisplaybreaks
\usepackage[implicit=false]{hyperref}
\usetikzlibrary{arrows.meta,
                backgrounds,
                chains,
                fit,
                quotes}
                \usepackage{tikz}
\usetikzlibrary{shapes,arrows}
\usepackage{circuitikz}
\usetikzlibrary{positioning}
\usepackage{amsmath}
\definecolor{fgreen}{RGB}{204,223,181}
\definecolor{fblue}{RGB}{85,181,255}
\definecolor{dblue}{RGB}{85,113,192}
\definecolor{dgreen}{RGB}{132,171,80} 
\definecolor{coralpink}{rgb}{0.97, 0.51, 0.47}

\usetikzlibrary{shapes.geometric}
\usetikzlibrary{decorations.pathreplacing,decorations.markings,calc}

\graphicspath{{images/}}
\newtheorem{theorem}{Theorem}

\newtheorem{lemma}[theorem]{Lemma}

\makeatletter
	{\end{proof}%
	\renewcommand{\qedsymbol}{\qedsymbol}}
\makeatother

\DeclareMathAlphabet\mathbfcal{OMS}{cmsy}{b}{n}

\definecolor{Gray}{gray}{0.9}
\definecolor{LightCyan}{rgb}{0.88,1,1}

\algnewcommand{\IIf}[1]{\State\algorithmicif\ #1\ \algorithmicthen}
\algnewcommand{\IELSEIF}[2]{\State\algorithmicelseif \ #2\ \algorithmicthen}
\algnewcommand{\IElse}[3]{\State\algorithmicelse\ #3\ }
\algnewcommand{\EndIIf}{\unskip\ \algorithmicend\ \algorithmicif}

\IEEEoverridecommandlockouts

\tikzstyle{Block} = [rectangle, rounded corners, minimum width=3cm, minimum height=1cm,text centered, draw=black, fill=red!30]

\usetikzlibrary{shapes.geometric}
\usetikzlibrary{decorations.pathreplacing,decorations.markings,calc}

\definecolor{Burgundy}{RGB}{10,0,0}

\usetikzlibrary{positioning}
\usetikzlibrary{arrows,decorations.markings}
\usetikzlibrary{arrows}
\usepackage[font=scriptsize,skip=5pt]{caption}

\usepackage{algorithm,algcompatible,amsmath}
\algnewcommand\INPUT{\item[\textbf{Input:}]}%
\algnewcommand\OUTPUT{\item[\textbf{Output:}]}%
\newcommand\numberthis{\addtocounter{equation}{1}\tag{\theequation}}

\begin{document}

\title{Optimizing Reported Age of Information with Short Error Correction and Detection Codes
\thanks{A shorter version of the paper has been submitted for review to IEEE Globecom workshops, 2023.}
}
\author{\IEEEauthorblockN{Sumanth S Raikar\IEEEauthorrefmark{1} and Rajshekhar  V Bhat\IEEEauthorrefmark{2}}
	
	\IEEEauthorblockA{Indian Institute of Technology Dharwad, Dharwad, Karnataka, India\\
		Email: 
		\IEEEauthorrefmark{1}211022005@iitdh.ac.in, 
		\IEEEauthorrefmark{2}rajshekhar.bhat@iitdh.ac.in  
}}

\maketitle
 
\begin{abstract}

Timely sampling and fresh information delivery are important in 6G communications. This is achieved by encoding samples into short packets/codewords for transmission, with potential decoding errors. We consider a broadcasting base station (BS) that samples information from multiple sources and transmits to respective destinations/users, using short-blocklength cyclic and deep learning (DL) based codes for error correction, and cyclic-redundancy-check (CRC) codes for error detection. We use a metric called reported age of information (AoI), abbreviated as RAoI, to measure the freshness of information, which increases from an initial value if the CRC reports a failure, else is reset. We minimize long-term average expected RAoI, subject to constraints on transmission power and distortion, for which we obtain age-agnostic randomized and age-aware drift-plus-penalty policies that decide which user to transmit to, with what message-word length and transmit power, and derive bounds on their performance. Simulations show that longer CRC codes lead to higher RAoI, but the RAoI achieved is closer to the true, genie-aided AoI. DL-based codes achieve lower RAoI.  Finally, we conclude that prior AoI optimization literature with finite blocklengths substantially underestimates AoI because they assume that all errors can be detected perfectly without using CRC.

\end{abstract}


\section{Introduction}
In upcoming sixth generation (6G)  semantic communications, fresh delivery of information packets is of significant importance  \cite{Book_Nikos,Book_Yates}. 
Accordingly, several metrics, including age of incorrect information (AoII) \cite{AoII}, age of incorrect estimates (AoIE)  \cite{Bhavya} and age of processed information (AoPI), 
  based on the concept of age of information (AoI), have been considered and optimized. For delivering information freshly, it may be useful to adopt codewords of short to ultra-short blocklengths, as the amount of time required to transmit them is low \cite{ISIT_2018_Oswaldo}. However, adopting short blocklength codewords results in non-negligible decoding errors, as they are unable to average out the effects of thermal noise, unlike  long codewords \cite{tail_letters}. 
Due to this, when short codewords are used, while the delay in transmitting a packet is small, the packet may not even be successfully delivered. 
For a codeword of given length, the error probability decreases if the message-word length is decreased or the transmit power is increased \cite{PPV}. However, decreasing the message-word length may result in distortion of information.
Moreover, existing work on AoI optimization literature does not consider the need for CRC to determine whether a transmission is successful or not. That is, they assume that the success or failure of a transmission is perfectly known without requiring any error detection mechanisms. 
 A widely used method to detect errors is to append cyclic redundancy check (CRC) code bits to the message bits, but they may not be able to detect all errors. 
  In practice, if CRC reports that a packet has passed the check, the packet is deemed successful. Otherwise, it is deemed a failure.
To account for this, we define a metric called reported AoI (RAoI), which is similar to AoI but evolves based on the CRC pass/fail report.  RAoI starts at zero, increases if CRC reports that a check has failed, and decreases otherwise. RAoI might differ from the true AoI, which only a genie at the receiver, possessing the exact transmitted message, can compute. 
We consider a base station (BS) that samples information from one of multiple sources and communicates over a broadcast channel to the respective destination/user using finite-blocklength cyclic and deep learning (DL) based  error correction   and CRC-based error detection codes. 
The objective is to optimize transmit power, message-word length, and user selection, to minimize long-term average RAoI, considering power and distortion constraints.
The key trade-off here is that using longer message-words results in lower distortion, but a higher probability of error if transmitted with lower power. On the other hand, adopting shorter message-words leads to higher distortion. 
\subsection{Literature Survey}
There are several works which consider optimization of  the AoI metric and its variants with finite-blocklength codewords. 
\cite{dual_hop_with_fbl} considers a dual-hop communication system and formulates an optimization problem for minimizing average AoIs by optimally allocating blocklengths. The formulated problem happens to be non-convex and the authors solve it suboptimally by approximating the average AoI by a convex function.
\cite{AoI_eh_with_time_iot} considers a scenario where the sensor harvests and collects energy for a certain duration and transmits status updates to a destination using the collected energy with finite blocklength codewords. The work optimizes the time duration of energy collected; the longer the duration for which the energy is collected, the higher is the energy collected on average and the higher is the energy with which the codewords can be transmitted thereby reducing the error rate at the cost of consuming more time. 
\cite{AoI_outage_fbl} studies AoI outage with respect to blocklength allocation in a two-user orthogonal multiple access system and develops a recursive algorithm to minimize AoI outage based on steady-state analysis of the corresponding Markov Decision Process. 
\cite{AoI_bec_arq} minimizes average AoI by finding the optimal blocklength, where a packet in error is retransmitted using automatic repeat request (ARQ). 
In \cite{AoI_outage_existance_downlink}, the authors characterize the behavior of average AoI with respect to the blocklength, in a scenario where a previously transmitted packet is retransmitted until success or until a new packet arrives at the transmitter;  the success or failure of a packet is assumed to be known. 
\cite{Modiano} minimizes average AoI  in a wireless sensor network with multiple users by optimizing on blocklength allocated to each user subject to a total blocklength constraint. 
\cite{tail_letters} considers an uplink communication system and minimizes the transmitter's average transmit power and blocklength subject to a constraint on the tail behavoir of the AoI of each sensors using a Lyapunov stochastic optimization framework. 
\cite{ISIT_2018_Oswaldo} studies delay and peak AoI minimization when the packets arrive according to an independent and identical distributed Bernoulli distribution and served using a first come first serve queue using finite-blocklength results from \cite{PPV}. 
\cite{hv_poor_stat_bound}  considers an AoI metric over a wireless channel and optimizes the transmit power so as to minimize the peak AoI violation probability, when communication is carried out using finite blocklength codewords. Other works which study optimization of AoI under short blocklength codewords include \cite{FB1,FB2,FB3,FB4}.  AoI optimization has also been considered in a variety network settings, including cognitive-radio networks \cite{CR-NOMA}, 
slotted ALOHA networks \cite{SlottedALOHA} and unmanned-aerial vehicle (UAV) networks \cite{UAV-Networks}.
A summary of previous works is provided in Table~\ref{literature_table}, which lists the optimization variables used.


 \begin{table}[t]
	\centering
 \vspace{0.2cm}
 \setlength{\tabcolsep}{0pt} 
\sisetup{group-digits=false} 
	\caption{A summary of decision variables in previous works on finite-blocklength communication for AoI optimization.
 }
 \small
	\begin{tabular}{|c|c|}
		\hline
		  Decision Variable  &  References  \\ \hline \hline 
            Blocklength  & \cite{dual_hop_with_fbl},\cite{AoI_outage_fbl},\cite{n_opt_MDP}    \\ \hline	
            Transmit power and blocklength & \cite{short_pkt_cognitive_iot}     \\ \hline
            Time duration of energy collected &   \cite{AoI_eh_with_time_iot}   \\ \hline
            Blocklength with ARQ  & \cite{AoI_bec_arq}     \\ \hline
            Blocklength for retransmissions \&	analysis of AoI behavior& \cite{AoI_outage_existance_downlink} \\
            \hline
            Blocklength allocation for multiple users & \cite{n_opt_MDP} \\
            \hline
            Transmitter's average transmit power and blocklength  & \cite{tail_letters}\\
            \hline
            Blocklength and queuing policy  & \cite{ISIT_2018_Oswaldo}\\
            \hline
            Blocklength and update rate & \cite{FB1} \\
            \hline

            
	\end{tabular}
 \medskip
\label{literature_table}
\end{table}

\subsection{Literature Gap and Contributions}
The key differences between current and earlier works in terms of system models, decision variables, and optimization objectives and  resulting challenges are as follows: 
\begin{itemize}[leftmargin=*]
\item \emph{Accounting for trade-off among message-word length, transmit power and distortion:}  Previous research has focused on optimizing channel blocklength, while message-word length optimization has been overlooked. In practice, varying message-word length is a common approach to achieving desired code rates with fixed channel blocklength \cite{dvb-s2x}. In this work, we consider message-word length optimization, which cannot be obtained by extending previous works on channel blocklength optimization. In freshness-critical applications, carrying forward remaining message bits to the next codeword may not be advantageous as they may become stale. Therefore, we may need to transmit fewer bits in the present time slot, leading to distortion, or transmit more bits at the expense of higher energy costs, which we take into account in our study, unlike in earlier works. 
\item \emph{Incorporation of practical error correction and detection codes:}  Previous studies on finite-blocklength codewords that make transmission decisions based on AoI have assumed perfect error detection.  
   However, this assumption is not realistic in practice as mechanisms such as CRC must be adopted for error detection  and they consume additional transmit power and channel resources. As a result, results obtained under the assumption of perfect error detection may not be accurate, especially for short block-lengths.  
We consider practical error correction and detection codes, which makes it difficult to derive an analytical expression for the error probability, in turn making it difficult to analytically optimize resource allocation. 
\item \emph{Consideration of broadcast channel:} Previous works on finite-blocklength communication have not considered broadcast channels (BC). We address this gap by considering scheduling users for transmission over a BC, in addition to optimizing message-word length and transmit power. At any given time, the BS can communicate with at most one user, as considered in \cite{Modiano, Kadota2018SchedulingPF,sched_algos_bcn}. While this simplifies communication compared to the case when multiple users are allowed to transmit simultaneously, scheduling and resource allocation remain non-trivial in this case, similar to \cite{Modiano, Kadota2018SchedulingPF,sched_algos_bcn}. 
In contrast to \cite{energy_eff_noma_bc}, which assumes that only the channel distribution information is available at the transmitter, our work considers the scenario where the instantaneous channel power gain information is available at the transmitter. This allows the transmitter to adapt the transmit powers to more effectively deliver bits to the users.
\end{itemize}


\begin{figure}[t]
    \centering
    \scalebox{1.2}{
    \begin{tikzpicture}[
state/.style={draw, rectangle,  minimum width=1cm, minimum height=.5cm},
arrow/.style={-Stealth, shorten >=1pt},
dot/.style={minimum size=4pt, inner sep=0pt, rounded corners=1pt, fill},
node distance=0.1cm and 2cm,
]
\node[state,fill=black!10] (x1) {\scriptsize Source $1$};
\node[state,fill=black!10, below=of x1] (x2) {\scriptsize Source $2$};
\node [below=-0.1cm of x2] {$\vdots$};
\node[state,fill=black!10, below=0.8cm of x2] (xn) {\scriptsize  Source $M$};

\node[state] at ($(x1)!0.5!(xn)+(2,0)$) (bs) {BS};

\node[state,fill=black!10, right=2.85cm of x1] (y1) {\scriptsize User $1$};
\node[state,fill=black!10, below=of y1] (y2) {\scriptsize  User $2$};
\node[state,fill=black!10, below=0.8cm of y2] (ym) {\scriptsize User $M$};

\draw[arrow] (x1) -- node[above, font=\small]{}(bs);
\draw[arrow] (x2) -- node[above, font=\small]{}(bs);
\draw[arrow] (xn) -- node[above, font=\scriptsize]{}(bs);

\node [below=-0.1cm of y2] {$\vdots$};

\draw[arrow, dotted] (bs) -- node[above, font=\small]{}(y1);
\draw[arrow, dotted] (bs) -- node[above, font=\small]{}(y2);
\draw[arrow, dotted] (bs) -- node[above, font=\small]{}(ym);

\end{tikzpicture}
    }
    \caption[]{We consider a broadcasting communication system with $M$ sources, a BS and $M$ destinations/users. The sources are co-located with the BS and the BS communicates with the users via a noisy channel.}
\label{fig:BS_SETUP}
\end{figure}
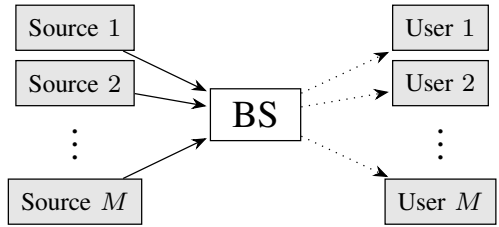

To fill the above  gap, 
we consider the problem of optimal allocation of message-word length, transmit power, and user scheduling in a  BC. We   minimize the long-term average RAoI, subject to average power and distortion constraints. 
Our main contributions are: 
\begin{itemize}[leftmargin=*]
    \item We consider practical cyclic and DL-based error correction and CRC-based error detection codes, and numerically obtain the relationship between error probability, message-word length, and transmit power.
    \item 
    We obtain age-agnostic   randomized  
  and age-aware drift-plus-penalty   policies for minimizing long-term average RAoI and derive bounds on their performance. 
    \item By  extensive simulations, we demonstrate that employing longer CRC codes results in higher RAoI due to their better error detection capability, and because they entail increased power and channel resource consumption. We also observe that DL-based codes achieve lower RAoI, as they can be designed with good performance for different message-word lengths unlike classical codes. We conclude that  existing finite-blocklength AoI optimization literature significantly underestimates actual AoI, as it assumes perfect error detection without adopting CRC.
\end{itemize}

\begin{figure*}[t]
 \centering   
\scalebox{0.65}{%
     \begin{tikzpicture}[auto,
    node distance = 18mm,
        start chain = A going right,
     block/.style = {rounded corners,
                     text width=#1, minimum height=12mm, align=center,
                     outer sep=0pt, on chain},
     block/.default = 1mm,
 container/.style = {rounded corners, fill=fgreen!50,
                     inner xsep=1mm, inner ysep=7mm},
                        ]
\centering
\node   [fill=black!10,block=20mm,draw=black,line width=1pt] {$\{0,1\}^{k_i^0}$};    
\node   [fill=black!10,block=15mm,draw=black,line width=1pt] {CRC};
\node   [fill=black!10,block=20mm,draw=black,line width=1pt] {Encoder};    
\node   [block=20mm,draw=black,line width=1pt] {Channel}; 
\node   [fill=black!10,block=20mm,draw=black,line width=1pt] {Decoder};    
\node   [fill=black!10,block=15mm,draw=black,line width=1pt] {CRC};
\node[above right=of A-6] (A-8) {$\Delta_i = n_i/N$};
\node[below right=of A-6] (A-9) {$\Delta_i = \Delta_i+1$};
%
       
\draw [-Stealth]
    (A-1) edge [""] (A-2)    
    (A-2) edge ["$k_i+c_i$"] (A-3)
    (A-3) edge ["$n_i$"] (A-4)
    (A-4) edge ["$n_i$"] (A-5);
     \draw[-Stealth] (A-5)  edge node [above,pos=0.5] {$k_i+c_i$}  (A-6);
     \draw[-Stealth] (A-5)  edge node [below,pos=0.5] {bits}  (A-6);
     \draw[-Stealth] (A-2)  edge node [below,pos=0.5] {bits}  (A-3);
     \draw[-Stealth] (A-3)  edge node [below,pos=0.5] {symbols}  (A-4);
     \draw[-Stealth] (A-4)  edge node [below,pos=0.5] {symbols}  (A-5);
\draw [-Stealth](A-6) -- ( $ (A-6.0)!0.2!(A-8.west|-A-6.0) $ ) |- (A-8.west) node[auto,pos=0.7] {Pass};
\draw [-Stealth](A-6) -- ( $ (A-6.0)!0.2!(A-9.west|-A-6.0) $ ) |- (A-9.west) node[auto,pos=0.7] {Fail};
    \end{tikzpicture}
}
\caption[]{The end-to-end communication setup considered in this work. The BS samples data packet of  length $k_i^0$ bits from  $i^{\rm th}$ source, compresses them to $k_i$ bits and appends  $c_i$ CRC bits  that are then encoded using either a classical or deep learning (DL) based error correction code. The encoded packets are transmitted to the user $i$ over a noisy channel. The user   decodes the received codeword and uses CRC to detect errors. If an error is detected, the packet is discarded and RAoI, $\Delta_i$ is increased. Otherwise, RAoI is reset.

}
\label{fig:Proposed-setup}
\end{figure*}
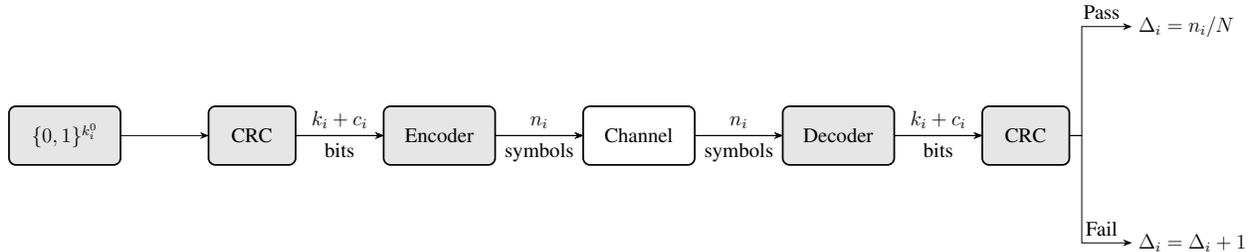

\section{System Model and Problem Formulation}
 We consider a broadcasting communication system with $M$ sources, a BS and $M$ destinations/users,   indexed by $i\in \mathcal{M}\triangleq\{1,2,\ldots,M\}$.  
 The system is time-slotted with unit-length slots which are indexed by $t\in \{1,2,\ldots\}$, where in a slot at most one user can transmit (as in \cite{Modiano}) a maximum of $N$ symbols. 
 At the start of slot $t$, a $k^0_i$-bit status update packet, $m_i(t)\in \{0,1\}^{k^0_i}$ can be generated at will from source $i$, attached to the BS. 
 The BS transmits $k_i(t)\leq k^0_i$ bits in slot $t$ by incurring certain distortion, which will be made concrete later in the section, where $k_i(t)\in \mathcal{K}_i\subset \mathbb{N}$.  
 The transmitter appends $c_i$-bit CRC code $\mathsf{D}$ to the $k_i(t)$ message bits, encodes the $(k_i(t)+c_i)$-bit CRC-appended message  using an encoder, $\mathsf{C}$, to obtain an $n_i$-symbol channel packet, $x_i(t)\in \mathbb{R}^{n_i}$, where $k^0_i+c_i\leq n_i\leq N$. Let $u_i(t)$ be the indicator variable which takes value $1$ if an encoded packet is transmitted in slot $t$ from user $i$ or zero otherwise.
 At slot $t$, let $P_i(t) \in \mathcal{P}_i\subset \mathbb{R}^+$ be the power allocated to a channel packet of user $i$ for transmission, i.e., $||x_i(t)||^2_2=nP_i(t)$, where $\mathcal{P}_i$ is a finite set. 
 The obtained packet,  $x_i(t)$, containing $n_i$ symbols, is transmitted over an additive white Gaussian noise (AWGN) BC, with noise variance corresponding to each receiver being unity. 
 The receiver decodes and obtains an estimate of the transmitted CRC-appended message sequence on which CRC check is performed, which may fail or pass. Let $\hat{m}_i(t)$ be the estimate of the message received in time slot $t$ from user $i$ and  $v_i(t)$ be the indicator variable which takes value of $1$ if the CRC passes at the receiver for a received codeword from user $i$ and zero otherwise.

 
We use the following metric for every user $i$: 
 \begin{align}\label{eq:age_evolution}
 \Delta_i(t) = \begin{cases}
 \Delta_{i}(t-1) + 1  & \text{if $u_i(t)v_i(t)=0$},\\
 \frac{n_i}{N} &\text{otherwise},
 \end{cases}
 \end{align}
 which we refer to as the instantaneous  RAoI at the end of slot $t$ in user $i$, for $t\in \{1,2,\ldots,\}$ and $i\in \{1,2,\ldots,M\}$ and $\Delta_i(0)\triangleq n_i/N$.    
The condition $u_i(t)v_i(t)=0$ may come true under the following two cases: (i) when a packet is transmitted, and the CRC fails i.e., when $u_i(t)=1$ and $v_i(t)=0$, or (ii) when  no packet is transmitted, i.e., when $u_i(t)=0$. Under the first case, the received packet is discarded and in both the cases, the RAoI increases. On the other hand, when the packet is transmitted and the CRC passes, we accept the packet and the RAoI is reset.  However, since $k_i(t)\leq k^0_i$, the packet is received with distortion of $d(k_i(t))$, where $d(\cdot)$ is a convex decreasing function.  That is, the distortion incurred in time slot $t$ for the packet from user $i$ is $u_i(t)v_i(t)d(k_i(t))$. We only count the distortion of successfully received packets.

Our objective is to deliver information as fresh as possible while keeping the average distortion and transmit power low. We achieve this by solving the following optimization problem for different choices of error detection and error correction codes,  $\mathsf{D}$ and $\mathsf{C}$, respectively:
\begin{subequations}\label{eq:main_opt_problem}
\begin{align}
 A^{\mathsf{D}, \mathsf{C}}_{\rm opt} = \min_{\pi}& \lim_{T\rightarrow \infty}\frac{1}{TM} \mathbb{E} \left[ \sum_{i=1}^M \sum_{t=1}^T  w_i\Delta_i(t) \right], \\
 \text{subject to}\ & \sum_{i=1}^{M} u_{i}(t) \leq 1, ~u_i(t)\in \{0,1\},\label{eq:main_policy_constraint}\\
 & \lim_{T\rightarrow \infty}\frac{1}{T}\sum_{t=1}^T \sum_{i=1}^M \mathbb{E}[u_i(t)P_i(t)]\leq \bar{P}, \label{eq:power_constraint}\\
 & \lim_{T\rightarrow \infty}\frac{1}{T}\sum_{t=1}^T  \mathbb{E}[u_i(t)v_i(t)d_i(k_i(t))]\leq \bar{d}_i, \label{eq:dist_constraint}\\
 & P_i(t) \in \mathcal{P}_i,\; k_i(t)\in \mathcal{K}_i, 
 \label{eq:kPset}
 \end{align}
 \end{subequations}
 for all $i\in \mathcal{M}$ and $t\in \{1,2,\ldots\}$, 
 where the expectations are with respect to the randomness in the channel and choice of policy $\pi$,  and the instantaneous RAoI, $\Delta_i(t)$ and distortion, $d_i(t)$.
 A policy $\pi$ gives a decision rule for selecting $u_i(t)$, $k_i(t)$ and $P_i(t)$ in slot $t$ for user $i$. 
Before solving  \eqref{eq:main_opt_problem}, we note the following lower bound, derivable via methods in \cite{eytan_modiano}:
\begin{subequations}\label{eq:lower_bound_opt_problem}
\begin{align}
L^{\mathsf{D}, \mathsf{C}}  = \min_{\pi}&\;\frac{1}{2M} \mathbb{E} \left[\sum_{i=1}^Mw_i\left( \frac{1}{\bar{q}_i}+\frac{n_i}{N}\right)\right], \\
 \text{subject to}\ &\eqref{eq:main_policy_constraint},\eqref{eq:power_constraint},\eqref{eq:dist_constraint},\eqref{eq:kPset},
 \end{align}
 \end{subequations}
where $\bar{q}_i \triangleq \lim_{T\rightarrow \infty}(1/T)\sum_{t=1}^T\mathbb{E} \left[u_i(t)v_i(t)\right]$ is the average number of  packets successfully delivered by  $i^{\rm th}$ source per slot under optimal solution to  \eqref{eq:lower_bound_opt_problem} and  $L^{\mathsf{D}, \mathsf{C}} \leq A^{\mathsf{D}, \mathsf{C}}_{\rm opt}$. 



\section{Solution}
In this section, we propose two policies: an age-agnostic stationary randomized policy (SRP) that does not depend on time or RAoI, and an age-aware drift-plus-penalty policy that operates using RAoI, and derive bounds on their performance.

 \subsection{Age-Agnostic Stationary Randomized Policy (SRP) }
The policy is: 
 \emph{User $i$ is selected in a time slot to transmit $k\in \mathcal{K}_i$ bits with power $P\in \mathcal{P}_i$, with probability  $\mu_i(k,P)$.} 
 


We now transform \eqref{eq:main_opt_problem} under the above SRP.  Let $\epsilon_i^{\mathsf{D}, \mathsf{C}}(k,P)$ be probability of successful reception of a packet, as indicated by the CRC. 
Then, the probability that a transmission  is successful in a slot is $p_i=\sum_{k\in \mathcal{K}_i,P\in \mathcal{P}_i}\epsilon_i^{\mathsf{D}, \mathsf{C}}(k,P)  \mu_i(k,P)$.
When a transmission is successful, the RAoI of user $i$ drops to $n_i/N$.   Hence,  the inter-delivery intervals form a sequence of independent and identical distributed random variables having  geometric distribution with parameter, $p_i$. In this case, the long-term average RAoI for user $i$, can be specialized to 
\begin{align*}
  \lim_{T\rightarrow \infty}\frac{1}{T}\sum_{t=1}^T  \mathbb{E}[w_i\Delta_i(t)]
 =w_i\left(\frac{1}{p_i}+\frac{n_i}{N}-1\right).  
\end{align*}

From this, we can rewrite \eqref{eq:main_opt_problem} as follows: 
\begin{subequations}\label{eq:SRP}
 \begin{align}
& A_{\rm SRP}^{\mathsf{D}, \mathsf{C}} =\min_{0 \leq \mu_i(k,P) \leq 1} \frac{1}{M} \  \sum_{i=1}^M w_i\left(\frac{1}{p_i}+\frac{n_i}{N}-1\right),  \numberthis \label{eq:SRPO}\\
 & \text{subject to}\; \sum_{i\in \{1,2,\ldots,M\}} \sum_{k\in \mathcal{K}_i,P\in\mathcal{P}_i}\mu_i(k,P) = 1,\label{eq:SRP_mu1} \\ 
 &\qquad \;\;\; \;\;\;\;\;\sum_{k\in \mathcal{K}_i,P\in \mathcal{P}_i}\sum_{i=1}^M  P\mu_i(k,P) \leq \bar{P},
 \label{eq:SRP_Power} \\
&\qquad  \;\;\; \;\;\;\;\;\sum_{k\in \mathcal{K}_i,P\in\mathcal{P}_i}\mu_i(k,P)\epsilon_i^{\mathsf{D}, \mathsf{C}}(k,P) d_i(k)\leq \bar{d}_i,\label{eq:SRPdist}
 & 
 \end{align}
 \end{subequations}
 for all $i\in \mathcal{M}$, $k\in \mathcal{K}_i$ and $P\in \mathcal{P}_i$. 

 The optimization problem in \eqref{eq:SRP} is a convex optimization problem because, the objective function is the inverse of an affine function with range in $\mathbb
{R}^{++}$ and all the constraint functions are affine. We solve the problem using standard numerical techniques.  We have the following result.


 
 \begin{theorem}
 	The optimal value $A_{\rm SRP}^{\mathsf{D}, \mathsf{C}} $ of problem in \eqref{eq:SRP} is at most twice of $A^{\mathsf{D}, \mathsf{C}}_{\rm opt}$ in \eqref{eq:main_opt_problem}, for any given  $\mathsf{D}$ and $\mathsf{C}$.
 \end{theorem}
 \begin{proof}
 Following \cite{TWC_RB}, we can show that a stationary randomized policy is an optimal solution to \eqref{eq:lower_bound_opt_problem}. The most general stationary randomized policy is the same as the one proposed above. Under this policy,  $\bar{q}_i = p_i$ and the constraints \eqref{eq:main_policy_constraint}, \eqref{eq:power_constraint} and  \eqref{eq:dist_constraint}  evaluate to \eqref{eq:SRP_mu1}, \eqref{eq:SRP_Power} and \eqref{eq:SRPdist}, respectively. Therefore, the constraints in \eqref{eq:SRP} and \eqref{eq:lower_bound_opt_problem} are identical and the objective function in \eqref{eq:lower_bound_opt_problem} is a scaled (by $1/2$) and shifted (by $+1$) form of that in \eqref{eq:SRP}. This means that both problems have the same optimal solution.  Hence,  ${A_{\rm SRP}^{\mathsf{D}, \mathsf{C}} }/{A^{\mathsf{D}, \mathsf{C}}_{\rm opt} } \leq {A_{\rm SRP}^{\mathsf{D}, \mathsf{C}} }/{L^{\mathsf{D}, \mathsf{C}}} <2$, as desired.
 
 \end{proof}


\subsection{Age-Aware Drift-Plus-Penalty  Policy (DPP)}
In the above SRP, the BS does not use the knowledge of instantaneous RAoI at the receiver.  In this section, we propose  a DPP policy, which utilizes the  knowledge of instantaneous RAoI at the receiver. 
Define the   virtual queues, $Q_1(t)$ and $\textbf{Q}_2(t)$, such that 
and $\mathbf{Q}_2(t) \triangleq [Q_{1,2}(t),Q_{2,2}\ldots,Q_{M,2}(t)]^T$. They evolve as follows: 
\begin{subequations}\label{eq:queue_evolution}
\begin{flalign}
&Q_{1}(t) = \max \{Q_{1}(t-1)-\bar{P}, 0\}+\sum_{i=1}^Mu_i(t) P_i(t),&& \label{avgpower_queue_evolution}\\
& Q_{i, 2}(t)=\max \left\{Q_{i, 2}(t-1)-\bar{d}_i, 0\right\}+u_i(t)v_i(t)d_i(t), &&\label{avgdistortion_queue_evolution}
& 
 \end{flalign}
\end{subequations}
where $Q_{1}(0)=0$ and $Q_{i,2}(0)=0$, for all $i \in \{1,2,\ldots,M\}$ and $t \in \{1,2,\ldots\}$.  
For computing \eqref{avgdistortion_queue_evolution}, we need to know whether or not the packet transmitted was decoded successfully using CRC, as indicated by $v_i(t)$.
Let $\textbf{S}(t-1)=\{Q_1(t-1),\mathbf{Q}_2(t-1), \boldsymbol{\Delta}(t-1)\}$, represent state at time instant $t-1$, where $\boldsymbol{\Delta}(t-1) \triangleq [\Delta_{1}(t-1),\ldots, \Delta_M(t-1)]$. 
 The conditional Lyapunov drift that measures the expected change in the Lyapunov function over time slots, given the system's state $\textbf{S}(t-1)$ in slot $t-1$ is given by
\begin{align*} 
	& \theta(\textbf{S}(t-1))=  \mathbb{E}\left[L\left(Q_{1}(t)\right)-L\left(Q_{1}(t-1)\right) \mid \textbf{S}(t-1)\right] \nonumber \\
	&\qquad \qquad \;\;\; + \mathbb{E}\left[L\left(\textbf{Q}_{2}(t)\right)-L\left(\textbf{Q}_{2}(t-1)\right) \mid \textbf{S}(t-1)\right],
\end{align*} 
where $L(Q_{1}(t))=(V_1/2) Q_{1}^2(t)$ and $L(\textbf{Q}_{2}(t))=(V_2/2)\sum_{i=1}^M Q_{i,2}^2(t)$. 
Using DPP approach in \cite{Neely}, a solution to \eqref{eq:main_opt_problem} is obtained by solving:
\begin{subequations}\label{eq:dpp_problem}
\begin{align}
  \text{minimize} & \  \theta(\textbf{S}(t-1))+  \sum_{i=1}^M \beta_i\mathbb{E}\left[\Delta_i(t)|\textbf{S}(t-1)\right],\label{eq:main_dpp_problem} \\
 \text{subject to}&\ \;  \sum_{i=1}^{M} u_{i}(t) \leq 1, u_i(t)\in \{0,1\},\label{eq:dpp_policy_constraint}
 \end{align}
 \end{subequations}
where $i\in \{1,2,\ldots,M\}$, $t\in\{1,2,\ldots,\}$. The parameters $\beta_1,\beta_2,\ldots, \beta_M>0$ act as trade-off factors, determining the relative importance of the penalty which is the weighted sum of expected RAoIs and the Lyapunov drift $\theta(\textbf{S}(t-1))$.
Instead of directly solving \eqref{eq:dpp_problem}, we derive an upperbound on its solution using the approach described in \cite{Neely}. We then demonstrate that by minimizing the upperbound, the average RAoI achieved is close to the optimal value $A_{\rm opt}^{\mathsf{D}, \mathsf{C}}$ specified in \eqref{eq:main_opt_problem}.
In Appendix A, we show that 
\begin{align}
&\theta(\textbf{S}(t-1)) \leq  MB_1 + V_1Q_1(t)\left(\sum_{i=1}^M \mathbb{E}\left[u_i(t)P_i(t)|\textbf{S}(t-1)\right]-\bar{P}\right)\nonumber\\
&+ MB_2 + \sum_{i=1}^M V_2Q_{i,2}(t)\left(\mathbb{E}\left[u_i(t)v_i(t)d_i(t)\mid \textbf{S}(t-1)\right]-\bar{d}_i\right),\nonumber
\end{align}
where $B_1 =  (V_1/2M) \left(\bar{P}^2 +  (\max \sum_{i=1}^M \mathcal{P}_i)^2\right)$ and $B_2 = (V_2/2M) \sum_{i=1}^M \left(\bar{d}_i^2 + (d_i^2(\min \mathcal{K}_i)\right)$, and $V_1>0$ and $V_2>0$. 
Using this,  we formulate the following problem:
\begin{subequations}\label{eq:dpp_upperbound_opt_problem}
\begin{align}
  \min_{u_i(t), k_i(t), P_i(t)}& \   V_1Q_{1}(t) \left( \sum_{i=1}^M u_i(t)P_i(t)-\bar{P}\right)+ \sum_{i=1}^M \beta_i \Delta_i(t) \nonumber \\ &   + V_2\sum_{i=1}^MQ_{i,2}(t)\left(u_i(t)\epsilon(k,P)d_i(t)-\bar{d}_i\right), \numberthis \label{eq:dpp_upper_bound_problem}\\
 \text{subject to}\ & \sum_{i=1}^{M} u_{i}(t) \leq 1, u_i(t)\in \{0,1\},\label{eq:dpp_upb_policy_constraint}
\end{align}
\end{subequations}
where 
$\Delta_i(t)=\left(\Delta_i (t-1)+1\right)\left(1-u_i(t)\right)+u_i(t)\left(\Delta_i(t-1)+1-\epsilon_i(k_i(t),P_i(t))\Delta_i(t-1)\right)$. We solve \eqref{eq:dpp_upperbound_opt_problem} as follows.  At each time slot $t$, we observe  $\textbf{S}(t-1)$ and select actions $u_i(t)$, $k_i(t)$, and $P_i(t)$ that minimize the objective in \eqref{eq:dpp_upperbound_opt_problem} in a greedy manner. We then update the virtual queues according to \eqref{eq:queue_evolution}.
The DPP policy (see Algorithm~\ref{algo:DPP}) obtained by solving \eqref{eq:dpp_upperbound_opt_problem} which is used to calculate the long-term average RAoI denoted by $A_{\rm DPP}^{\mathsf{D}, \mathsf{C}}$. 
We have the following result.



\begin{algorithm}[t] 
	\caption{DPP Policy}
	 \label{algo:DPP}
	\begin{algorithmic}[1]
		\STATE \textbf{Initialization:} $t =1$, $Q_1(1) = Q_{i,2}(1)=0$ for all $i\in \{1,2,\ldots,M\}$. 
		\STATE \textbf{Input:}  $\mathcal{K}_i,\mathcal{P}_i, \beta_i, \epsilon_i(k,P), \bar{P} \; \text{and} \; \bar{d}_i \;\; \forall i$ and  $T\gg 1$.
		\WHILE {$t \leq T$}
		\STATE Choose the action $(u_i(t), k_i(t), P_i(t))$ that minimizes \eqref{eq:dpp_upperbound_opt_problem}, such that $u_i(t)$ satisfies \eqref{eq:dpp_policy_constraint}, $k_i(t)\in \mathcal{K}_i$, and $P_i(t) \in \mathcal{P}_i$. 
	    \STATE Update $Q_1(t+1), Q_{i,2}(t+1)$ for all $i\in \{1,2,\ldots,M\}$ using \eqref{eq:queue_evolution}.
		\STATE  $t \leftarrow t+1$ 
		\ENDWHILE
	\end{algorithmic}
\end{algorithm}
 
 
\begin{theorem}\label{thm:DPP}
The DPP policy in Algorithm~\ref{algo:DPP} satisfies the average power and distortion constraints in \eqref{eq:power_constraint} and \eqref{eq:dist_constraint}, respectively. Moreover,  
\begin{align}
    \frac{A_{\rm DPP}^{\mathsf{D}, \mathsf{C}}}{A_{\rm opt}^{\mathsf{D}, \mathsf{C}}} \leq 2 + \frac{1}{L^{\mathsf{D}, \mathsf{C}}}\left(2(B_1+B_2)-\frac{1}{MN}\sum_{i=1}^M w_i n_i\right), 
\end{align}
where   $A_{\rm DPP}^{\mathsf{D}, \mathsf{C}}$ is the optimal RAoI achieved under DPP. 
\end{theorem}
\begin{proof}
We follow the techniques adopted in \cite{Modiano} and \cite{Neely} to prove the theorem. 
We first upper-bound the conditional Lyapunov drift of the virtual queues $Q_1(t)$ and $\textbf{Q}_2(t)$ and then prove that they are strongly stable.  Subsequently, we relate the minimum long-term average RAoI achieved by DPP and a lower bound on the optimal RAoI in \eqref{eq:main_opt_problem} to obtain the result.   Refer to   Appendix B   for detailed explanations. 
\end{proof}

\subsection{Periodic Round-Robin (PRR) Policy}
In this policy, the BS sends information to a user after a fixed number of slots, known as the period. We set a fixed value for the message-word length $k$ and sequentially select transmit powers from available choices, ensuring the chosen transmission period meets the average power requirement. Each user's information is transmitted with a fixed offset relative to a user selected as a reference user.


\section{Simulation Results}
In this section, we study how the average RAoI varies with system parameters for different choices of $\mathsf{D}$ and $\mathsf{C}$.


\subsection{Obtaining $\epsilon_i^{\mathsf{D}, \mathsf{C}}(k,P)$ for different $\mathsf{D}$ and $\mathsf{C}$}
For each $\mathsf{D}$ and $\mathsf{C}$, we compute $\epsilon_i^{\mathsf{D}, \mathsf{C}}(k,P)$, for $k\in \{4,5,\ldots,11\}$ and $P\in \{1,2,3,4\}$, and  utilize it to solve the RAoI minimization problem. We consider $n_i=15$ for all $i\in \{1,2,\ldots,M\}$. 

\subsubsection{Error Detection Codes, $\mathsf{D}$}
We consider the following error detecting codes, using CRC: 
Given a message represented as a polynomial $m(x)$, where the coefficients correspond to the bits of the message, the CRC-encoded message is obtained as:
 $p(x)=x^cm(x)+(x^cm(x)\mod g(x)$), where $g(x)$ is the generator polynomial. The message is then transmitted over a channel. At the receiver, the received message is decoded as: $\hat{p}(x)=p(x)+e(x)$,  where $e(x)$ is the error polynomial. The error is detected by checking the remainder,  $\hat{p}(x)\mod g(x)$. If the remainder is zero, then there was no error. Otherwise, an error occurred. Note that if $e(x)$ happens to be a multiple of $g(x)$, the error cannot be detected. In this study, we consider two CRC polynomials: a $1$-bit CRC polynomial, $g(x) = x+1$, and a $3$-bit CRC polynomial, $g(x) = x^3 + x+1$.  A detailed description of CRCs and their performance can be found in \cite{koopman_best_crc_site}.
\subsubsection{Error Correction Codes, $\mathsf{C}$}
We consider classical cyclic and DL-based codes for communication. Below, we describe how a codebook is generated for a given codeword length, $n$ and message-word length, $k$.
\paragraph{Classical Cyclic Codes}
We obtain binary codewords for classical cyclic codes using the \verb|cyclpoly| function in \textsc{Matlab}, and modulate them using binary phase shift keying. The \verb|cyclpoly| function takes in $n$ and $K$, and outputs a generator polynomial for a linear block code, using which codebook can be generated. The generator polynomial and codeword length is then passed to \verb|cyclgen|, which gives the corresponding parity-check matrix. Using \verb|syndtable|, a syndrome table is generated for decoding received symbols via syndrome-decoding \cite{cyclpoly_matlab}. 
\paragraph{DL-based Codes}
We obtain DL-based codes  using an auto-encoder architecture, consisting of an encoder, a noise layer and a decoder, as in \cite{timoshea}. 
 A one-hot encoded message, $m$,  of length $2^k$, is encoded using, $f_{\theta}(m))$ to obtain a codeword, $x\in \mathbb{R}^n$, satisfying $||x||_2^2\leq nP$. The codeword $x$ is then passed through a Gaussian noise layer with $n$ nodes, each adding a zero-mean unit-variance Gaussian random variable.  Subsequently, the decoder $g_{\phi}(\cdot)$, processes the output of the noise layer to obtain an estimate, $\hat{m}$, of the message $m$, i.e., $\hat{m} = g_{\theta}(h(f_{\theta}(m)))$. We parameterize $f_{\theta}$ and $g_{\phi}$ using fully-connected feed-forward neural networks and minimize the categorical cross-entropy loss, $l(m, \hat{m})$, between $m$ and $\hat{m}$. 
 The encoder  comprises of three layers: the first layer maps $\mathbb{R}^{2^k} \to \mathbb{R}^{n}$ using ReLU activation, while the second layer maps $\mathbb{R}^n \to \mathbb{R}^n$ with linear activation followed by a normalization layer to ensure $||x||_2^2\leq nP$. The decoder has two layers: the initial layer maps $\mathbb{R}^n \to \mathbb{R}^{2^k}$ using ReLU activation, and the second layer maps $\mathbb{R}^{2^k} \to \mathbb{R}^{2^k}$, with soft-max activation.


\begin{table}[t]
	\centering
 \vspace{0.3cm}
	\setlength{\tabcolsep}{2pt} 
	\sisetup{group-digits=false} 
	\caption{Comparison of average RAoI under different policies}
 \small
	\begin{tabular}{|p{4cm}|>{\centering\arraybackslash}p{2.5cm}|>{\centering\arraybackslash}p{2.5cm}| >
     {\centering\arraybackslash}p{2.5cm}|}
		\hline
		Error Correction and Detection Codes & Periodic Round-Robin & Stationary Randomized & Drift-Plus-Penalty \\ 
		\hline\hline 
		Cyclic code (genie) & $2.46$ & $2.14$ & $1.605$ \\ 
		\hline	
		Cyclic code (CRC-$1$) & $1.92$ & $2.07$ & $1.555$ \\ 
		\hline
		DL code (genie) & $2.74$ & $2.03$ & $1.529$ \\ 
		\hline
		DL code (CRC-$1$) & $1.96$ & $2.02$ & $1.516$ \\ 
		\hline
		PPV error expression & $1.84$ & $2$ & $1.506$ \\
		\hline
	\end{tabular}
	\medskip
	\label{srp_rr_comparison_table}
\end{table}

 \subsection{Studying Impact of System Parameters on Average RAoI}

In our study, we examine a scenario with two sources and destination nodes, i.e., with $M=2$.

\subsubsection{Comparing RAoI Across Policies}
In Table~\ref{srp_rr_comparison_table}, we compare the average age obtained using periodic round-robin and stationary randomized policies for  $w_1=w_2=1$, $\bar{P}=2$ units and  $\bar{d}_1 = \bar{d}_2 = 2^{-0.1}$. In the round-robin policy, a packet is transmitted to a user every $2$ slots. To satisfy the average power and distortion constraints, we set $k=10$ and alternate the value of $P$ among $\{1,2,3,4\}$ across the slots. Unlike the stationary randomized policy, the round-robin policy makes decisions independent of the error expression. When the genie is used, the error expression reflects the true errors, and the stationary randomized policy is optimized accordingly. However, when CRC is used, not all errors may be reported, affecting the  optimization. As a result, the stationary randomized policy performs better than the round-robin policy in the genie-aided case, and vice versa in the CRC case. The DPP (with $\beta_1=\beta_2 = 100$)  consistently outperforms both the round-robin and randomized policies by leveraging instantaneous information to adapt its decisions based on transmission success or failure.

\begin{figure*}[t!]
     \centering
     \begin{subfigure}{0.45\textwidth}
         \centering
         \scalebox{0.9}{\input{aoi_plots_cyclic_ae_crc/AOI_POWER_AE_HAMM_PPV_CRC1_label_updated_distorion0.99}}
         \caption{Comparison of the optimal average RAoI under SRP   with respect to $\bar{P}$, assuming $\bar{d}_1=\bar{d}_2=0.99$.}
         \label{fig:AoI_vs_avg_power_SRP}
     \end{subfigure}
     \hfill
     \begin{subfigure}{0.45\textwidth}
         \centering
         \scalebox{0.9}{\input{aoi_plots_cyclic_ae_crc/DPP_AoI_POWER_AE_CYCLIC_PPV_CRC1_label_updated_fixed_power1}}
         \caption{Comparison of the  optimal average RAoI under DPP with respect to $\bar{P}$, assuming $\bar{d}_1=\bar{d}_2=0.99$.}
         \label{fig:AoI_vs_avg_power_DPP}
     \end{subfigure}
     \medskip
     \begin{subfigure}{0.45\textwidth}
         \centering
         \scalebox{0.9}{\input{aoi_plots_cyclic_ae_crc/AoI_DISTORTION_AE_CYCLIC_PPV_CRC1_label_updated_fixed_power1}}
         \caption{Comparison of the optimal average RAoI in \eqref{eq:SRP} with respect to $\bar{d}$, assuming $\bar{d}_1=\bar{d}_2=\bar{d}$ and $\bar{P}=2$.}
         \label{fig:AoI_vs_avg_distortion_SRP}
     \end{subfigure}
     \hfill
     \begin{subfigure}{0.45\textwidth}
         \centering
         \scalebox{0.9}{\input{aoi_plots_cyclic_ae_crc/DPP_AoI_DISTORTION_AE_CYCLIC_PPV_CRC1_label_updated_fixed_power1}}
         \caption{Comparison of the optimal average RAoI in \eqref{eq:dpp_upperbound_opt_problem} with respect to $\bar{d}$, assuming $\bar{d}_1=\bar{d}_2=\bar{d}$ and $\bar{P}=2$.}
         \label{fig:AoI_vs_avg_distortion_DPP}
     \end{subfigure}
     
    \caption{Comparison of the optimal average  RAoI with respect to bounds on average power and distortion. The word ``genie'' signifies perfect packet error detection at the destination.}
    \label{fig:AoI_power_and_distortion_graph}
\end{figure*}
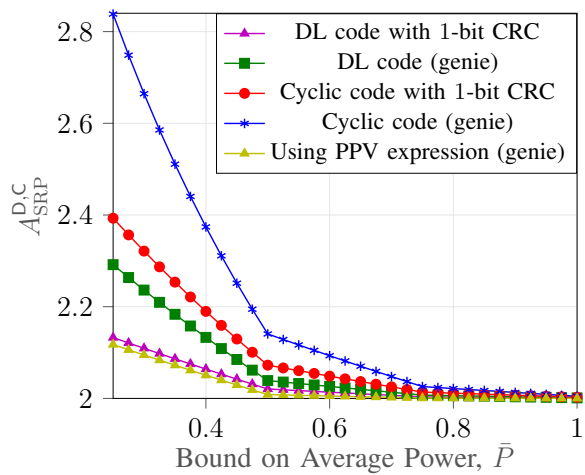
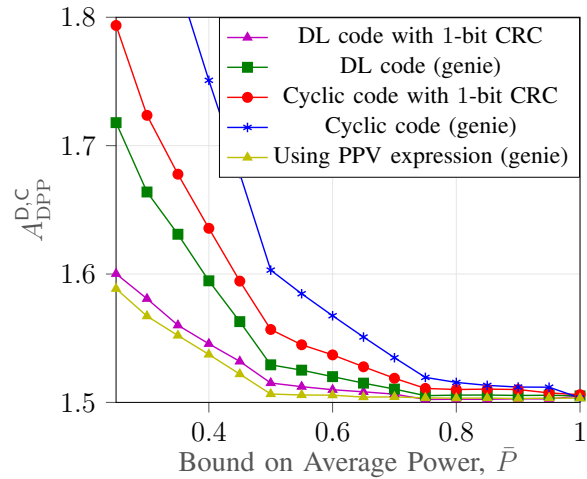
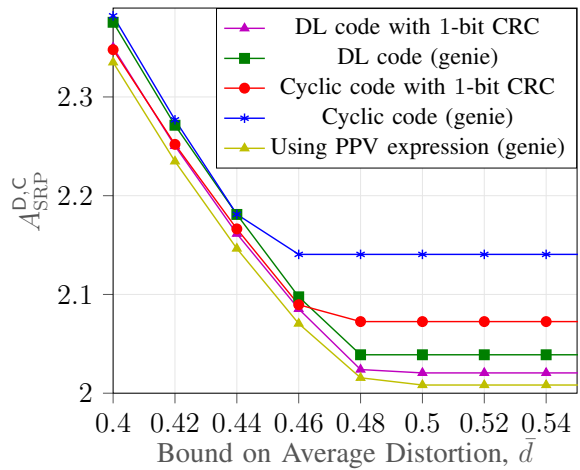
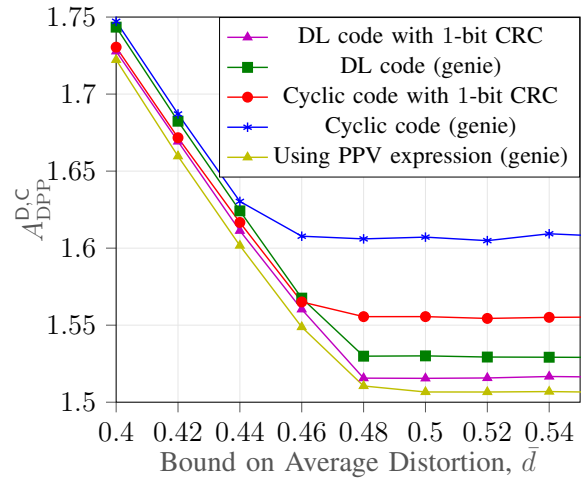
\subsubsection{Variation of Optimal Average RAoI with Power and Distortion Bounds}
Fig.~\ref{fig:AoI_power_and_distortion_graph} compares the optimal average RAoI in \eqref{eq:SRP} and \eqref{eq:dpp_problem} with respect to bounds on average power,  $\bar{P}$ and distortion,  $\bar{d}_1=\bar{d}_2=\bar{d}$.  From  Fig.~\ref{fig:AoI_vs_avg_power_SRP} and Fig.~\ref{fig:AoI_vs_avg_power_DPP}, we observe a decrease in RAoI as the value of $\bar{P}$ increases. 
This is because, higher $\bar{P}$
enables more frequent transmissions, and allocates more power for transmitting a fixed number of bits, reducing error probability. 
Due to this, the packets are successfully delivered frequently, reducing RAoI.  
In SRP case  in Fig.~\ref{fig:AoI_vs_avg_power_SRP}, the minimum achievable RAoI remains at $2$ even with unlimited power. This is due to the fact that with unlimited power, the transmissions can occur in every slot, and due to the  constraint of non-simultaneous transmission by both users, the transmission can occur with probability $0.5$ in each user (all the user parameters are considered to be identical for obtaining the result) and almost all transmissions will be successful, as transmissions occur with high power. With this, the RAoI can be computed to be $2$. 
However, in DPP, the minimum achievable RAoI goes to  $1.5$ at higher powers. This is because when the bound on the power is very large, it is optimal to transmit to users in an alternating manner, once every two slots.
Similarly, from Fig.~\ref{fig:AoI_vs_avg_distortion_SRP} and Fig.~\ref{fig:AoI_vs_avg_distortion_DPP}, we note a decrease in RAoI with increasing $\bar{d}$. This is due to the fact that as $\bar{d}$ increases, the number of transmitted message bits decreases. Transmitting fewer bits reduces the probability of error for a given power, potentially allowing for lower transmit power per transmission and increased transmission frequency, resulting in low RAoI.

Furthermore, based on Fig.~\ref{fig:AoI_power_and_distortion_graph}, we draw the following inferences. Including a $1$-bit CRC polynomial in any error correcting code, whether DL-based or classical cyclic code, leads to a lower reported RAoI. 
This is because, the $1$-bit CRC is unable to detect all errors and may incorrectly indicate absence of errors even when errors are present. Consequently, when the CRC indicates absence of errors, the RAoI is reset, falsely assuming successful delivery of the packet. Furthermore, in our adopted method, DL codes jointly design coding and modulation, whereas cyclic codes use binary phase shift keying modulation, which has fewer degrees of freedom. Hence, the RAoI achieved by the DL codes is lower. Utilizing the expression from \cite{PPV}, referred to as PPV, yields lower $\epsilon(k, P)$ values compared to DL and cyclic codes, resulting in the lowest computed RAoI.  From this, we can conclude that the  existing finite-blocklength AoI optimization literature, especially those which use the bounds in PPV \cite{PPV},  significantly underestimate actual AoI, as they assume perfect error detection without adopting CRC.

  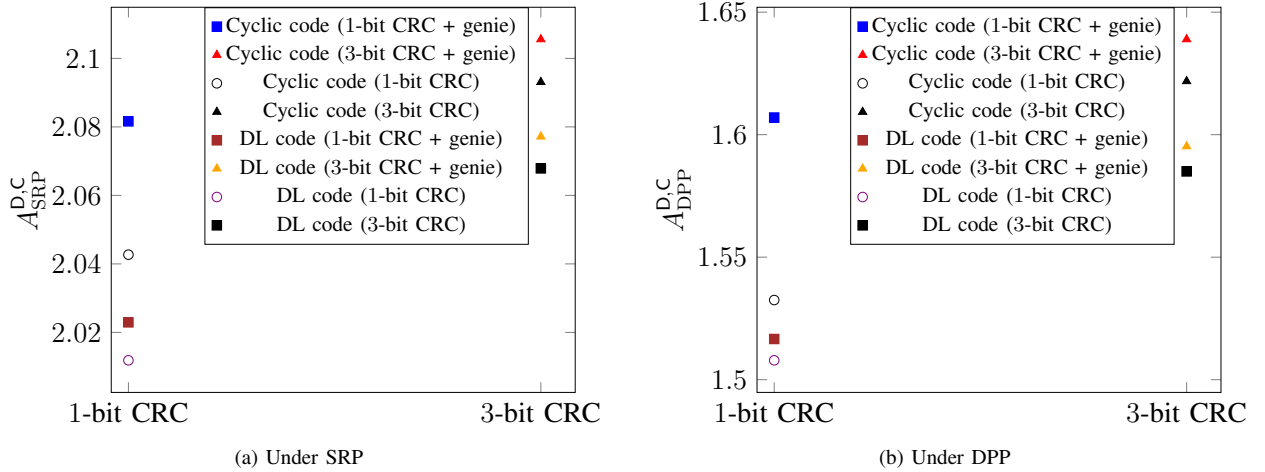
\begin{figure*}[t]
     \centering
     \begin{subfigure}{0.48\textwidth}
         \centering
         \scalebox{1}{\input{aoi_plots_cyclic_ae_crc/test}}
         \caption{Under SRP}
         \label{fig:SRP-CRC}
     \end{subfigure}
     \hfill
     \begin{subfigure}{0.48\textwidth}
         \centering
         \scalebox{1}{\input{aoi_plots_cyclic_ae_crc/crc1_vs_crc3_dpp}}
         \caption{Under DPP}
         \label{fig:DPP-CRC}
     \end{subfigure}
    \caption{Optimal average RAoI achieved with $1$-bit and $3$-bit CRCs. We observe that longer CRC codes lead to increased RAoI, as they require more power and channel resources, and because they have better error detection ability.}
    \label{fig:AoI_vs_crc1_crc3_plot}
\end{figure*}

 \begin{figure*}[t]
     \centering
     \begin{subfigure}{0.45\textwidth}
         \centering
         \scalebox{0.85}{\input{aoi_plots_cyclic_ae_crc/AoI_POWER_DISTORTION_AE_CRC}}
         \caption{Contour plot showing RAoI variation under SRP for DL code with $1$-bit CRC, for different values of $\bar{P}$ and $\bar{d}_1=\bar{d}_2=\bar{d}$.}
         \label{fig:AoI_cyclic_p_d_contour_SRP}
     \end{subfigure}
     \hfill
     \begin{subfigure}{0.45\textwidth}
         \centering
         \scalebox{0.85}{\input{aoi_plots_cyclic_ae_crc/AoI_POWER_DISTORTION_DPP_AE_CRC1_updated}}
         \caption{Contour plot showing RAoI variation under DPP for DL code with $1$-bit CRC, for different values of $\bar{P}$ and $\bar{d}_1=\bar{d}_2=\bar{d}$.}
         \label{fig:AoI_cyclic_p_d_contour_DPP}
     \end{subfigure}
    \caption{Average RAoI for different values of power and distortion bounds}
    \label{fig:AoI_cyclic_p_d_contour}
\end{figure*}
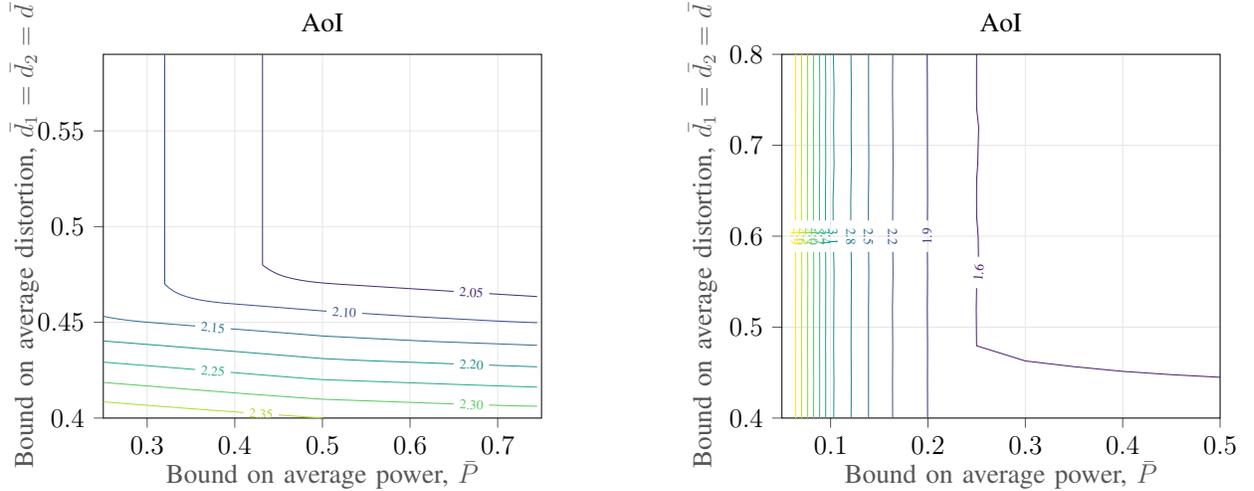
\subsubsection{Optimal average RAoI with $1$-bit and $3$-bit CRCs}
In Fig.~\ref{fig:AoI_vs_crc1_crc3_plot}, we compare the optimal average RAoI achieved under SRP and DPP using $1$-bit and $3$-bit CRCs. In the legend, ``$1$-bit CRC + genie'' refers to the case where CRC bits are transmitted and assumed to detect all errors. This is done to enable a fair comparison with the case where CRC bits are used, while also considering the power allocation to the CRC bits in the genie aided case. From Fig.~\ref{fig:SRP-CRC} and Fig.~\ref{fig:DPP-CRC}, we  observe that as the number of CRC bits increases, the average RAoI also increases. However, importantly, we observe a decrease in the gap between the reported RAoI and true genie-aided AoI.

\subsubsection{Average RAoI for different power and distortion bounds}
Fig. \ref{fig:AoI_cyclic_p_d_contour} presents a contour plot illustrating the optimal achievable average AoI under \eqref{eq:SRP} for varying average power, $\bar{P}$, and distortion, $\bar{d}$. As observed in Fig. \ref{fig:AoI_power_and_distortion_graph}, the AoI decreases when $\bar{d}$ is increased for a fixed $\bar{P}$, and vice versa. Moreover, we observe that maintaining the same AoI requires increasing average distortion, $\bar{d}$ when $\bar{P}$ is decreased, and vice versa.


\section{Conclusions and Future Work}
In this paper, we considered a problem of optimal allocation of message-word length, transmit power, and user scheduling in a  broadcast channel. We  minimized the long-term reported AoI (RAoI), subject to average power and distortion constraints, and a constraint that information can be transmitted to at most a single user at any time. 
We made the following contributions:
(i) we considered practical cyclic and DL-based error correcting   and CRC-based error detecting codes, and numerically obtained the relationship between error probability, message-word length, and transmit power,  (ii) using which we obtained age-agnostic SRP and age-aware DPP policies for scheduling and resource allocation, and (iii) we carried out extensive simulations and derived meaningful inferences, including the following: DL-based codes   achieve a lower RAoI than classical codes because they can be designed with decent performance for any message-word length, while classical codewords may not be able to achieve good performance for all lengths.
Longer CRC codes lead to increased RAoI, as they require more power and channel resources, and because they have better error detection ability. We finally conclude that the existing AoI literature underestimates true AoI by assuming perfect error detection without CRC.

\section{APPENDIX}
\subsection{Derivation of an Upperbound on Drift, $\theta(\textbf{S}(t-1))$}
Below, we obtain an upper bound on drift $\theta(\textbf{S}(t-1))$. 
Consider
\begin{flalign}
&L(Q_{1}(t))-L(Q_{1}(t-1))=\frac{V_1}{2} \left(Q_{1}^2(t)-Q_{1}^2(t-1)\right) \nonumber \\
& =\frac{V_1}{2}\left(\max \left\{Q_{1}(t-1)-\bar{P}, 0\right\}+ \sum_{i=1}^M u_i(t) P_i(t)\right)^2\nonumber \\
&\quad -\frac{V_1}{2} Q_{1}^2(t-1) \nonumber \\
& \overset{(a)}{\leq} \frac{V_1}{2}\left(\bar{P}^2+\left(\sum_{i=1}^M u_i(t) P_i(t)\right)^2\right) + V_1Q_{1}(t-1)\left(\sum_{i=1}^M u_i(t) P_i(t)-\bar{P}\right), \label{upb_power}
\end{flalign}
where 
(a) holds true because $V_1$ is a strictly positive quantity,  $Q_{1}(t),\bar{P},P_i(t)\geq 0$ and due to the following inequality $(\max\{a-b, 0\}+c)^2\leq  a^2+b^2+c^2+2a(c-b)$, for any $a,b,c\geq 0$. Taking expectation on both sides of \eqref{upb_power}, we have
\begin{flalign}
    &\mathbb{E}\left[L\left(Q_{1}(t)\right)-L\left(Q_{1}(t-1)\right) \mid \textbf{S}(t-1)\right] \nonumber\\ &\leq 
    \frac{V_1}{2}\bar{P}^2 + \frac{V_1}{2}\mathbb{E}\left[\left(\sum_{i=1}^Mu_i(t) P_i(t)\right)^2|\textbf{S}(t-1)\right] \nonumber \\ &\;\;\;\;+ V_1Q_{1}(t-1) \sum_{i=1}^M\mathbb{E}\left[u_i(t) P_i(t)|\textbf{S}(t-1)\right]-\bar{P} \nonumber \\
    &\leq MB_1 -  V_1Q_{1}(t-1)\left(\bar{P}-\sum_{i=1}^M\mathbb{E}\left[ u_i(t) P_i(t)|\textbf{S}(t-1)\right]\right), \label{upb_drift_power}
\end{flalign}
where $B_1 = {V_1}/{(2M)}  \left(
\bar{P}^2 + (\max \sum_{i=1}^M u_i(t)P_i(t))^2\right)$.
Along the similar lines, we can obtain the following: 
\begin{flalign}
    &\mathbb{E}\left[L\left(\textbf{Q}_{2}(t)\right)-L\left(\textbf{Q}_{2}(t-1)\right) \mid \textbf{S}(t-1)\right] \nonumber\\
    & \leq \frac{V_2}{2}\sum_{i=1}^M \bar{d}_i^2+ \frac{V_2}{2}\mathbb{E}\left[(u_i(t)v_i(t)d_i(t))^2|\textbf{S}(t-1)\right] \nonumber\\
     & +V_2 Q_{i, 2}(t-1)\mathbb{E}\left[(u_i(t)v_i(t)d_i(t)-\bar{d}_i)\mid \textbf{S}(t-1)\right] \nonumber \\
    &\leq MB_2 - \sum_{i=1}^M V_2Q_{i,2}(t-1)\mathbb{E}\left[\bar{d}_i-u_i(t)\epsilon(k_i(t),P_i(t))\mid \textbf{S}(t-1)\right], \label{upb_drift_dist} 
\end{flalign}
where $B_2 = {V_2}/{(2M)} \sum_{i=1}^M \bar{d}_i^2 + (\max\{u_i(t)v_i(t)d_i(t)\})^2$ and  via the law of iterated expectations, we have,
\begin{align}
    &\mathbb{E}_{A,V}\left[V_2Q_{i, 2}(t-1)\left(u_i(t)v_i(t)d_i(t)-\bar{d}_i\right) \mid \textbf{S}(t-1)\right] \nonumber \\
    &\overset{(a)}{=}\mathbb{E}_A\left[\mathbb{E}_V\left[V_2Q_{i, 2}(t-1)\left(u_i(t)v_i(t)d_i(t)-\bar{d}_i\right) \right.\right. \nonumber \\
    & \quad \left. \left.| u_i(t),k_i(t),P_i(t) \right]|\textbf{S}(t-1)\right]\nonumber \\
    &= V_2Q_{i,2}(t-1)\left(\mathbb{E}_A\left[u_i(t)\epsilon(k_i(t),P_i(t))d_i(t)\mid \textbf{S}(t-1)\right]-\bar{d}_i\right),\nonumber
\end{align}
where $A$ represents the instantaneous (possibly) random action and $V$ is the indicator random variable representing the success or failure of a packet, and (a) follows because \newline  $\mathbb{E}_V\left[v_i(t)\mid u_i(t), k_i(t),P_i(t)\right]=\epsilon(k_i(t),P_i(t))$. Using the bounds obtained in \eqref{upb_drift_power} and \eqref{upb_drift_dist}, 
we can then write the following upper-bound: 
\begin{align}
   &\theta(\textbf{S}(t-1))\leq M(B_1+B_2) \nonumber \\ 
   &-  V_1Q_{1}(t-1)\left(\bar{P}-\mathbb{E}\left[\sum_{i=1}^M u_i(t) P_i(t)|\textbf{S}(t-1)\right] \right) \nonumber\\
& -\sum_{i=1}^M V_2Q_{i,2}(t-1)\left(\bar{d}_i-\mathbb{E}[u_i(t)\epsilon(k_i(t),P_i(t))d_i(t)\mid \textbf{S}(t-1)]\right). \label{upper_bound_dpp_drift}
\end{align}

\subsection{Proof of Theorem~\ref{thm:DPP}}
To prove the first part of the theorem, which says that the DPP policy satisfies average power and distortion constraints, it is sufficient to prove the below lemma \cite{Neely}. 
\begin{lemma}\label{lma:strong_stability_lemma}
	The virtual queues $Q_{1}(t)$ and $Q_{i,2}(t)$  are strongly stable, i.e.,
	\begin{align} \label{queue-stability}
            \lim_{T \to \infty}\frac{1}{T}\sum_{t=1}^{T} \mathbb{E}[Q_{1}(t-1)]<\infty \text{ and }
		\lim_{T \to \infty}\frac{1}{T}\sum_{t=1}^{T}\mathbb{E}[Q_{i,2}(t-1)]<\infty, \;\; \forall i \in \{1,2,\dots, M\}.   
	\end{align}
\end{lemma}
\begin{proof}

From \eqref{upper_bound_dpp_drift} and \eqref{eq:main_dpp_problem}, we have
\begin{align}
   &{\theta}(\textbf{S}(t-1))+   \sum_{i=1}^M \beta_i \mathbb{E}\left[{\Delta}_i(t)|\textbf{S}(t-1)\right]  \nonumber\\
   &\leq M(B_1+B_2)- V_1Q_{1}(t-1)\left(\bar{P}-\sum_{i=1}^M\mathbb{E}\left[u_i(t){P}_i(t)|\textbf{S}(t-1)\right]\right)\nonumber\\
   &-\sum_{i=1}^M V_2Q_{i,2}(t-1)\left(\bar{d}_i-\mathbb{E}\left[u_i(t)\epsilon(k_i(t),P_i(t))d_i(t)|\textbf{S}(t-1)\right]\right) +  \sum_{i=1}^M \beta_i \mathbb{E}\left[\Delta_i(t)|\textbf{S}(t-1)\right]. \label{upb_q1_q2}
\end{align}
The DPP policy in Algorithm~\ref{algo:DPP}  minimizes the upperbound  \eqref{upb_q1_q2} on the DPP term, ${\theta}(\textbf{S}(t-1))+   \sum_{i=1}^M \beta_i \mathbb{E}\left[{\Delta}_i(t)|\textbf{S}(t-1)\right]$. Hence, a stationary randomized policy $\omega$, that  selects User $i$ with power $P$ and message-word length $k$,  with probability $\mu_i(k, P)$ in slot $t$, we have, 
\begin{align}
    &\sum_{i=1}^M\mathbb{E}\left[ u_i(t){P}_i(t)|\textbf{S}(t-1)\right] \leq \sum_{i=1}^M \sum_{k,P}P\mu_i(k,P),\nonumber \\    &\mathbb{E}\left[u_i(t)\epsilon(k_i(t),P_i(t))d_i(t)\mid \textbf{S}(t-1)\right] \leq  \sum_{k,P}d_i(k)\epsilon_{i}(k,P)\mu_i(k,P),\nonumber \\
    & \sum_{i=1}^M \beta_i \mathbb{E}\left[\Delta_i(t)|\textbf{S}(t-1)\right] \leq  \sum_{i=1}^M \beta_i \mathbb{E}\left[{\Delta}^\omega_i(t)\right]. 
    \label{dpp_srp_relation}
\end{align} 
We also have, 
\begin{align}\label{eq:expectedAoI}
    & \mathbb{E}\left[{\Delta}^\omega_i(t)\right]=   \Delta_i(t-1)+1-\Delta_i(t-1)\sum_{k,P}\mu_i(k,P) \epsilon_i(k,P).
\end{align}

From \eqref{dpp_srp_relation}, \eqref{eq:expectedAoI} and \eqref{upb_q1_q2}, we get, 
\begin{align}
    &{\theta}(\textbf{S}(t-1))+  \sum_{i=1}^M \beta_i \mathbb{E}\left[{\Delta}_i(t)|\textbf{S}(t-1)\right]\nonumber \\
   &\leq M(B_1+B_2) \nonumber - V_1Q_{1}(t-1)\left(\bar{P}-\sum_{i=1}^M \sum_{k,P}P\mu_i(k,P) \right) \nonumber\\
& -V_2\sum_{i=1}^M Q_{i,2}(t-1)\left(\bar{d}_i- \sum_{k,P}d_i(k)\epsilon_{i}(k,P)\mu_i(k,P)\right)\nonumber \\
&+   \sum_{i=1}^M \beta_i \left(\Delta_i(t-1)+1-\Delta_i(t-1)\sum_{k,P}\mu_i(k,P) \epsilon_i(k,P)\right). \nonumber
\end{align}
Summing either  side from $t=1$ to $t=T$,  dividing by $TM$, rearranging the terms and taking expectation with respect to $\mathbf{S}(t-1)$, we get, $LHS_1+LHS_2 \leq RHS_1+RHS_2+RHS_3$, where
\begin{align}
    &LHS_1 = \frac{V_1}{TM}\left(\bar{P}-\sum_{i=1}^M \sum_{k,P}P\mu_i(k,P)\right)\sum_{t=1}^{T}\mathbb{E}[Q_{1}(t-1)] \nonumber \\ 
   &+\frac{V_2}{TM}\sum_{i=1}^M \left(\bar{d}_i-  \sum_{k,P}d_i(k)\epsilon_{i}(k,P)\mu_i(k,P)\right)\sum_{t=1}^{T} \mathbb{E}[Q_{i,2}(t-1)] \nonumber \\
    &LHS_2 = \frac{1}{2M}\sum_{i=1}^M \beta_i \sum_{k,P}\mu_i(k,P) \epsilon_i(k,P) \frac{1}{T}\sum_{t=1}^{T} \mathbb{E}\left[\Delta_i(t-1)\right]\nonumber \\
    &RHS_1 = (B_1+B_2)+\frac{1}{2M}\sum_{i=1}^M \beta_i \nonumber \\
    &RHS_2 = \frac{V_1}{2TM}  \mathbb{E}\left[Q_{1}^2(0)-Q_{1}^2(T)\right]+\frac{V_2}{2TM} \sum_{i=1}^M \mathbb{E}\left[Q_{i,2}^2(0)-Q_{i,2}^2(T)\right]\nonumber \\
    &RHS_3 = \frac{1}{2TM}\sum_{i=1}^M \beta_i \mathbb{E}\left[\Delta_i(0)-\Delta_i(T)\right].
    \nonumber
\end{align}
 Recall that  $\Delta_i(0)=n_i/N$, $Q_{1}(0)=\textbf{Q}_{2}(0)=0$. 
Hence,   we have $RHS_2+RHS_3 \leq 0 \nonumber$.
Noting that $LHS_2 \geq 0$, we have, $LHS_1 \leq RHS_1$.  Since both the terms in $LHS_1$ are positive and $RHS_1 < \infty$, we can conclude that each queue is strongly stable i.e., \eqref{queue-stability} holds true. 
\end{proof}

We now   prove the second part of Theorem~\ref{thm:DPP}. 
\begin{proof}
Consider the inequality $LHS_1 + LHS_2 \leq RHS_1+RHS_2+RHS_3$. Since $RHS_2+RHS_3\leq 0$ and $LHS_1\geq 0$, we have $LHS_2\leq RHS_1$. That is, 
\begin{align}
    &\frac{1}{2M}\sum_{i=1}^M \beta_i\sum_{k,P}\mu_i(k,P) \epsilon_i(k,P) \lim_{T\rightarrow \infty}\frac{1}{T}\sum_{t=1}^{T} \left(\mathbb{E}\left[\Delta_i(t-1)\right]\right)\nonumber \\
    &\leq (B_1+B_2)+\frac{1}{2M}\sum_{i=1}^M \beta_i. \label{bound_inequality}
\end{align}
Let $q_i^L(k, P)$ be the average number of successful packets using $k$ as the messageword length and $P$ as the transmit power, for the policy that solves the lower bound problem in \eqref{eq:lower_bound_opt_problem}. Clearly, the average number of successful packets, $q_i^L = \sum_{k,P}q_i^L(k,P)$. 
Evaluating the expression for lower bound we get
\begin{align}
    L^{\mathsf{D}, \mathsf{C}}  = \frac{1}{2M} \left(\sum_{i=1}^M w_i\left( \frac{1}{q_i^L}+\frac{n_i}{N}\right)\right). \label{lower_bound_equation}
\end{align}
Assigning
$ \mu_i(k,P) = q_i^L(k,P)/\epsilon_i(k,P)$ and $\beta_i = 2w_i/q_i^L$, and noting that \newline
$\sum_{i=1}^M q_i^L(k,P)/\sum_{k,P}\epsilon_i(k,P) \leq 1$, from \eqref{bound_inequality}, we get
\begin{align}
    &A_{\rm DPP}^{\mathsf{D}, \mathsf{C}} \leq 2(B_1+B_2) + \frac{1}{M}\sum_{i=1}^M \frac{w_i}{q_i^L}. \label{dpp_bound_inequality}
\end{align}
Comparing \eqref{lower_bound_equation} and \eqref{dpp_bound_inequality} gives us the following expression
\begin{align}
    &A_{\rm DPP}^{\mathsf{D}, \mathsf{C}}\leq 2L^{\mathsf{D}, \mathsf{C}}+\left(2(B_1+B_2)-\frac{1}{M}\sum_{i=1}^M \frac{w_i n_i}{N}\right).
\end{align}
Therefore, we have,  \begin{align}
    \frac{A_{\rm DPP}^{\mathsf{D}, \mathsf{C}}}{A_{\rm opt}^{\mathsf{D}, \mathsf{C}}} \leq 2 + \frac{1}{L^{\mathsf{D}, \mathsf{C}}}\left(2(B_1+B_2)-\frac{1}{MN}\sum_{i=1}^M w_i n_i\right).
\end{align}
This follows due to the fact that $L^{\mathsf{D}, \mathsf{C}}\leq A_{\rm opt}^{\mathsf{D}, \mathsf{C}}.$
\end{proof}

\balance

\bibliographystyle{ieeetr}
\bibliography{references}

\end{document}

%% file: aoi_plots_cyclic_ae_crc/AOI_POWER_AE_HAMM_PPV_CRC1_label_updated_distorion0.99.tex
\begin{tikzpicture}

\definecolor{darkviolet1910191}{RGB}{191,0,191}
\definecolor{dimgray85}{RGB}{85,85,85}
\definecolor{gainsboro229}{RGB}{229,229,229}
\definecolor{goldenrod1911910}{RGB}{191,191,0}
\definecolor{green01270}{RGB}{0,127,0}

\definecolor{c1}{RGB}{191,191,191}
\definecolor{c2}{RGB}{0,0,0}
\definecolor{c3}{RGB}{0,127,100}

\begin{axis}[
axis background/.style={},
axis line style={black},
tick align=outside,
tick pos=left,
x grid style={gray!20},
 grid=both,
grid style={line width=.1pt, draw=gray!10},
major grid style={line width=.2pt,draw=gray!50},
xlabel=\textcolor{dimgray85}{Bound on Average Power, $\bar{P}$},
xmin=0.25, xmax=1,
xtick style={color=dimgray85},
y grid style={gray!20},
ylabel=\textcolor{dimgray85}{$A_{\rm SRP}^{\mathsf{D}, \mathsf{C}}$},
ymin=2, ymax=2.84,
ytick style={color=dimgray85},
legend style={at={(1,1)},anchor=north east}
]
\addplot [semithick, darkviolet1910191, mark=triangle*, mark size=2, mark options={solid}] 
table[x expr=\thisrowno{0}*0.5, y index=1] {%
0.5 2.13266224650181
0.55 2.12089358906002
0.6 2.10925409597271
0.65 2.09774167932643
0.7 2.0863542568764
0.75 2.07508972117049
0.8 2.06394622172901
0.85 2.05292176297822
0.9 2.04201444698096
0.95 2.031222488351
1 2.02054388010021
1.05 2.0187995994744
1.1 2.01705826715722
1.15 2.01531996358804
1.2 2.01358470224268
1.25 2.01185233508916
1.3 2.01012299676581
1.35 2.008396622154
1.4 2.00667320973866
1.45 2.00495274496825
1.5 2.00323524659272
1.55 2.00296990627284
1.6 2.00270465377856
1.65 2.00243947508143
1.7 2.00217436271113
1.75 2.00190932202399
1.8 2.00164435139763
1.85 2.00137948237637
1.9 2.00111464211489
1.95 2.00084987728415
2 2.00058521126726
2.05 2.00058517154209
2.1 2.00058517128024
2.15 2.00058517122729
2.2 2.00058517119691
2.25 2.00058517117809
2.3 2.0005851711742
2.35 2.00058517117056
2.4 2.00058517116797
2.45 2.00058517116952
2.5 2.00058517117929
2.55 2.00058517119224
2.6 2.00058517120803
2.65 2.00058517122618
2.7 2.00058521205836
2.75 2.00058520888479
2.8 2.00058520609141
2.85 2.00058520365936
2.9 2.00058520135941
2.95 2.00058519919512
};
\addlegendentry{\small DL code with $1$-bit CRC}
\addplot [semithick, green01270, mark=square*, mark size=2, mark options={solid}]
table[x expr=\thisrowno{0}*0.5, y index=1] {%
0.5 2.29173827845079
0.55 2.26366936951106
0.6 2.23627972419399
0.65 2.20954434535042
0.7 2.18344186957732
0.75 2.15794834189936
0.8 2.13304325847836
0.85 2.10870646089979
0.9 2.08491874239483
0.95 2.06166170498345
1 2.03891785589069
1.05 2.03570058676054
1.1 2.0324934720916
1.15 2.02929644458395
1.2 2.02610945855373
1.25 2.02293246088434
1.3 2.01976541117253
1.35 2.01660824123235
1.4 2.01346094553047
1.45 2.01032351034074
1.5 2.00719579346866
1.55 2.00658963739071
1.6 2.00598385067363
1.65 2.00537842636338
1.7 2.00477336571118
1.75 2.00416866921421
1.8 2.00356433998248
1.85 2.00296037486242
1.9 2.00235677523753
1.95 2.00175353930048
2 2.00115066223129
2.05 2.00115066110591
2.1 2.00115066044092
2.15 2.0011506595682
2.2 2.00115065944487
2.25 2.00115065960133
2.3 2.00115065964325
2.35 2.00115065970529
2.4 2.00115065969933
2.45 2.00115065972579
2.5 2.00115065971385
2.55 2.00115065972507
2.6 2.00115065970121
2.65 2.00115065973688
2.7 2.00115065969429
2.75 2.0011506597391
2.8 2.00115065973117
2.85 2.00115065977234
2.9 2.00115065978971
2.95 2.00115065978049
};
\addlegendentry{\small  DL code (genie)}
\addplot [semithick, red, mark=*, mark size=2, mark options={solid}]
table[x expr=\thisrowno{0}*0.5, y index=1] {%
0.5 2.39306012150808
0.55 2.35661122173003
0.6 2.32125601480647
0.65 2.28694599626177
0.7 2.25363539791286
0.75 2.22128120117136
0.8 2.18984296790133
0.85 2.15928215030049
0.9 2.12956258687909
0.95 2.10065003369337
1 2.0725120419927
1.05 2.06647917483984
1.1 2.06048131462729
1.15 2.05451817922119
1.2 2.04858947224829
1.25 2.0426948805927
1.3 2.03683411044856
1.35 2.03100687659284
1.4 2.02521290842338
1.45 2.01945187469418
1.5 2.01372356172249
1.55 2.01260190744372
1.6 2.01148153663305
1.65 2.0103624157421
1.7 2.00924453757686
1.75 2.00812790465132
1.8 2.00701251369077
1.85 2.00589834554098
1.9 2.00478542312022
1.95 2.00367373633881
2 2.00256328344586
2.05 2.00256328131896
2.1 2.00256328535895
2.15 2.00256328934239
2.2 2.00256329238529
2.25 2.00256329331217
2.3 2.00256329724928
2.35 2.00256329889711
2.4 2.0025633001316
2.45 2.00256330135575
2.5 2.00256330226228
2.55 2.00256330308326
2.6 2.00256330388905
2.65 2.00256330519601
2.7 2.00256330491937
2.75 2.00256330547481
2.8 2.00256330682755
2.85 2.00256330724435
2.9 2.00256330671703
2.95 2.00256330793353
};
\addlegendentry{\small Cyclic code with $1$-bit CRC}
\addplot [semithick, blue, mark=asterisk, mark size=2, mark options={solid}]
table[x expr=\thisrowno{0}*0.5, y index=1] {%
0.5 2.83826834936244
0.55 2.74867804610903
0.6 2.6645704767502
0.65 2.58545742527699
0.7 2.51090674284282
0.75 2.44053483618591
0.8 2.37399990322535
0.85 2.31099653469879
0.9 2.25125083033804
0.95 2.19451644138074
1 2.14057136243542
1.05 2.12851324237999
1.1 2.11659025765996
1.15 2.10480008216308
1.2 2.09314055564924
1.25 2.08160949049686
1.3 2.07020478355173
1.35 2.05892435388302
1.4 2.04776618383718
1.45 2.0367283160568
1.5 2.02580880276256
1.55 2.02366933409817
1.6 2.02153439275229
1.65 2.01940394822644
1.7 2.01727799153241
1.75 2.01515649395298
1.8 2.01303946076887
1.85 2.010926880549
1.9 2.00881871465324
1.95 2.00671496837588
2 2.00461562616777
2.05 2.00461562361287
2.1 2.00461562615514
2.15 2.00461562764067
2.2 2.00461562747583
2.25 2.00461558521573
2.3 2.00461562593112
2.35 2.00461562312418
2.4 2.00461562124117
2.45 2.00461562107373
2.5 2.00461562096712
2.55 2.00461562085542
2.6 2.00461562078457
2.65 2.00461562071412
2.7 2.00461562070433
2.75 2.00461562063838
2.8 2.00461562061916
2.85 2.0046156206143
2.9 2.00461562062468
2.95 2.00461562062024
};
\addlegendentry{\small  Cyclic code (genie)}

\addplot [semithick, goldenrod1911910, mark=triangle*, mark size=2, mark options={solid}]
table[x expr=\thisrowno{0}*0.5, y index=1] {%
0.5 2.11749978596494
0.55 2.10604627071507
0.6 2.09471600336992
0.65 2.08350691685757
0.7 2.07241724910752
0.75 2.06144501373282
0.8 2.0505882420091
0.85 2.03984541338461
0.9 2.02921440538161
0.95 2.01869368251906
1 2.00828148576931
1.05 2.00752168635253
1.1 2.0067624620105
1.15 2.00600381348712
1.2 2.00524572438852
1.25 2.00448823247927
1.3 2.00373129986428
1.35 2.00297495177028
1.4 2.00221915082644
1.45 2.00146394287117
1.5 2.0007092957264
1.55 2.00064553840414
1.6 2.0005818107075
1.65 2.00051804051628
1.7 2.00045430582872
1.75 2.00039057419175
1.8 2.00032682496597
1.85 2.00026311922256
1.9 2.0001993896937
1.95 2.00013564435519
2 2.00007195591854
2.05 2.00007192241562
2.1 2.00007194436024
2.15 2.00007193218107
2.2 2.00007193404795
2.25 2.00007193423035
2.3 2.00007193402581
2.35 2.00007193368926
2.4 2.00007193359155
2.45 2.00007193412782
2.5 2.00007193575034
2.55 2.00007193731559
2.6 2.00007193583296
2.65 2.00007193310467
2.7 2.00007193351345
2.75 2.00007193239523
2.8 2.00007193256317
2.85 2.00007193717532
2.9 2.00007193165485
2.95 2.00007193290448
};
\addlegendentry{\small Using PPV expression (genie)}


\end{axis}

\end{tikzpicture}

%% file: aoi_plots_cyclic_ae_crc/DPP_AOI_POWER_AE_CYCLIC_PPV_CRC1_label_updated_fixed_power1.tex
\begin{tikzpicture}

\definecolor{darkviolet1910191}{RGB}{191,0,191}
\definecolor{dimgray85}{RGB}{85,85,85}
\definecolor{gainsboro229}{RGB}{229,229,229}
\definecolor{goldenrod1911910}{RGB}{191,191,0}
\definecolor{green01270}{RGB}{0,127,0}

\begin{axis}[
axis background/.style={},
axis line style={black},
tick align=outside,
tick pos=left,
x grid style={gray!20},
 grid=both,
grid style={line width=.1pt, draw=gray!10},
major grid style={line width=.2pt,draw=gray!50},
xlabel=\textcolor{dimgray85}{Bound on Average Power, $\bar{P}$},
xmin=0.25, xmax=1,
xtick style={color=dimgray85},
y grid style={gray!20},
ylabel=\textcolor{dimgray85}{$A_{\rm DPP}^{\mathsf{D}, \mathsf{C}}$},
ymin=1.5, ymax=1.8,
ytick style={color=dimgray85},
legend style={at={(1,1)},anchor=north east}
]
\addplot [semithick, darkviolet1910191, mark=triangle*, mark size=2, mark options={solid}]
table[x expr=\thisrowno{0}*0.5, y index=1] {%
0.1 5.58176
0.2 3.157285
0.3 2.31562
0.4 1.894375
0.5 1.600175
0.6 1.580745
0.7 1.560105
0.8 1.54566
0.9 1.531935
1 1.515115
1.1 1.51227
1.2 1.50988
1.3 1.508245
1.4 1.506445
1.5 1.50232
1.6 1.502465
1.7 1.502435
1.8 1.502775
1.9 1.5027
2 1.504955
};
\addlegendentry{\small DL code with $1$-bit CRC}
\addplot [semithick, green01270, mark=square*, mark size=2, mark options={solid}]
table[x expr=\thisrowno{0}*0.5, y index=1] {%
0.1 5.976835
0.2 3.37165
0.3 2.46924
0.4 2.014785
0.5 1.717915
0.6 1.663905
0.7 1.630985
0.8 1.59472
0.9 1.56293
1 1.529285
1.1 1.52513
1.2 1.520075
1.3 1.515025
1.4 1.51039
1.5 1.505265
1.6 1.50569
1.7 1.50577
1.8 1.50541
1.9 1.505575
2 1.50512
};
\addlegendentry{\small DL code (genie)}
\addplot [semithick, red, mark=*, mark size=2, mark options={solid}]
table[x expr=\thisrowno{0}*0.5, y index=1] {%
0.1 6.239295
0.2 3.50092
0.3 2.553775
0.4 2.0938
0.5 1.79363
0.6 1.723625
0.7 1.677765
0.8 1.63562
0.9 1.59447
1 1.5568
1.1 1.544895
1.2 1.53702
1.3 1.52776
1.4 1.518825
1.5 1.510905
1.6 1.510045
1.7 1.510355
1.8 1.51003
1.9 1.50739
2 1.50601
};
\addlegendentry{\small Cyclic code with $1$-bit CRC}

\addplot [semithick, blue, mark=asterisk, mark size=2, mark options={solid}]
table[x expr=\thisrowno{0}*0.5, y index=1] {%
0.1 7.27636
0.2 4.09655
0.3 2.96499
0.4 2.428875
0.5 2.13021
0.6 1.950255
0.7 1.8274
0.8 1.75088
0.9 1.678245
1 1.603045
1.1 1.584645
1.2 1.56751
1.3 1.55092
1.4 1.53482
1.5 1.519485
1.6 1.51557
1.7 1.513275
1.8 1.511965
1.9 1.511845
2 1.503515
};
\addlegendentry{\small Cyclic code (genie)}

\addplot [semithick, goldenrod1911910, mark=triangle*, mark size=2, mark options={solid}]
table[x expr=\thisrowno{0}*0.5, y index=1] {%
0.1 5.53717
0.2 3.12917
0.3 2.302075
0.4 1.879995
0.5 1.58857
0.6 1.56719
0.7 1.55211
0.8 1.53735
0.9 1.522055
1 1.506535
1.1 1.50576
1.2 1.505645
1.3 1.50414
1.4 1.5043
1.5 1.50353
1.6 1.503275
1.7 1.503415
1.8 1.503005
1.9 1.50337
2 1.503265
};
\addlegendentry{\small Using PPV expression (genie)}
\end{axis}

\end{tikzpicture}

%% file: aoi_plots_cyclic_ae_crc/AOI_DISTORTION_AE_CYCLIC_PPV_CRC1_label_updated_fixed_power1.tex
\begin{tikzpicture}

\definecolor{darkviolet1910191}{RGB}{191,0,191}
\definecolor{dimgray85}{RGB}{85,85,85}
\definecolor{gainsboro229}{RGB}{229,229,229}
\definecolor{goldenrod1911910}{RGB}{191,191,0}
\definecolor{green01270}{RGB}{0,127,0}

\begin{axis}[
axis background/.style={},
axis line style={black},
tick align=outside,
tick pos=left,
x grid style={gray!20},
 grid=both,
grid style={line width=.1pt, draw=gray!10},
major grid style={line width=.2pt,draw=gray!50},
xlabel=\textcolor{dimgray85}{Bound on Average Distortion, $\bar{d}$},
xmin=0.4, xmax=0.55,
xtick style={color=dimgray85},
y grid style={gray!20},
ylabel=\textcolor{dimgray85}{$A_{\rm SRP}^{\mathsf{D}, \mathsf{C}}$},
ymin=2, ymax=2.39,
ytick style={color=dimgray85},
legend style={at={(1,1)},anchor=north east}
]
\addplot [semithick, darkviolet1910191, mark=triangle*, mark size=2, mark options={solid}]
table {%
0.4 2.35007016796038
0.42 2.25037020817966
0.44 2.16148903985004
0.46 2.08534343603288
0.48 2.02394788667139
0.5 2.02054388085303
0.52 2.02054388375992
0.54 2.02054387966622
0.56 2.02054388014649
0.58 2.02054388201893
0.6 2.02054393492559
0.62 2.0205438930977
0.64 2.02054394308393
0.66 2.02054390169154
0.68 2.02054388922171
0.7 2.0205438886824
0.72 2.02054387343499
0.74 2.02054388225667
0.76 2.02054388373545
0.78 2.02054388033879
};
\addlegendentry{\small DL code with $1$-bit CRC}

\addplot [semithick, green01270, mark=square*, mark size=2, mark options={solid}]
table {%
0.4 2.37559787712026
0.42 2.27134598332333
0.44 2.18101987625957
0.46 2.09760284927999
0.48 2.03891783953929
0.5 2.03891784399751
0.52 2.03891784928633
0.54 2.0389178468126
0.56 2.03891779823969
0.58 2.03891770081105
0.6 2.03891783816449
0.62 2.0389178338585
0.64 2.03891784409697
0.66 2.03891785187707
0.68 2.03891788635347
0.7 2.03891781796645
0.72 2.03891785175851
0.74 2.03891785742176
0.76 2.03891786499244
0.78 2.03891787541491
};
\addlegendentry{\small DL code (genie)}

\addplot [semithick, red, mark=*, mark size=2, mark options={solid}]
table {%
0.4 2.34785905471695
0.42 2.25196512829064
0.44 2.16638155059011
0.46 2.08961333679801
0.48 2.07251200670939
0.5 2.07251182404052
0.52 2.07251186953167
0.54 2.07251200776663
0.56 2.07251200434259
0.58 2.07251188948522
0.6 2.07251204189031
0.62 2.07251203737738
0.64 2.07251201616294
0.66 2.07251202565847
0.68 2.07251203854371
0.7 2.07251205361831
0.72 2.07251207229486
0.74 2.07251209383753
0.76 2.07251201536956
0.78 2.07251201367913
};
\addlegendentry{\small Cyclic code with $1$-bit CRC}

\addplot [semithick, blue, mark=asterisk, mark size=2, mark options={solid}]
table {%
0.4 2.38213814728847
0.42 2.27681397209042
0.44 2.18114089043521
0.46 2.14057130739321
0.48 2.14057131325517
0.5 2.14057132580013
0.52 2.14057131705983
0.54 2.14057135530139
0.56 2.14057135765611
0.58 2.14057134120409
0.6 2.14057132161735
0.62 2.1405712395049
0.64 2.14057131368301
0.66 2.14057121811755
0.68 2.14057127394095
0.7 2.14057139301608
0.72 2.14057133829108
0.74 2.14057136959427
0.76 2.14057132251526
0.78 2.14057132594468
};
\addlegendentry{\small Cyclic code (genie)}

\addplot [semithick, goldenrod1911910, mark=triangle*, mark size=2, mark options={solid}]
table {%
0.4 2.33503026693551
0.42 2.23475908465901
0.44 2.14646499751855
0.46 2.07046458348882
0.48 2.01559453828626
0.5 2.0082814701883
0.52 2.00828147873614
0.54 2.00828148134739
0.56 2.00828148718181
0.58 2.00828148116381
0.6 2.00828148057149
0.62 2.00828149832859
0.64 2.00828148757987
0.66 2.0082815182578
0.68 2.00828148419255
0.7 2.00828148639027
0.72 2.00828148906638
0.74 2.00828149227206
0.76 2.00828149561125
0.78 2.00828149843947
};
\addlegendentry{\small Using PPV expression (genie)}



\end{axis}

\end{tikzpicture}

%% file: aoi_plots_cyclic_ae_crc/DPP_AOI_DISTORTION_AE_CYCLIC_PPV_CRC1_label_updated_fixed_power1.tex
\begin{tikzpicture}

\definecolor{darkviolet1910191}{RGB}{191,0,191}
\definecolor{dimgray85}{RGB}{85,85,85}
\definecolor{gainsboro229}{RGB}{229,229,229}
\definecolor{goldenrod1911910}{RGB}{191,191,0}
\definecolor{green01270}{RGB}{0,127,0}

\begin{axis}[
axis background/.style={},
axis line style={black},
tick align=outside,
tick pos=left,
x grid style={gray!20},
 grid=both,
grid style={line width=.1pt, draw=gray!10},
major grid style={line width=.2pt,draw=gray!50},
xlabel=\textcolor{dimgray85}{Bound on Average Distortion, $\bar{d}$},
xmin=0.4, xmax=0.55,
xtick style={color=dimgray85},
y grid style={gray!20},
ylabel=\textcolor{dimgray85}{$A_{\rm DPP}^{\mathsf{D}, \mathsf{C}}$},
ymin=1.5, ymax=1.75,
ytick style={color=dimgray85},
legend style={at={(1,1)},anchor=north east}
]
\addplot [semithick, darkviolet1910191, mark=triangle*, mark size=2, mark options={solid}]
table {%
0.4 1.72775
0.42 1.669235
0.44 1.6113
0.46 1.56028
0.48 1.51561
0.5 1.51549
0.52 1.51573
0.54 1.51668
0.56 1.516275
0.58 1.51641
0.6 1.516125
0.62 1.515665
0.64 1.51529
0.66 1.51563
0.68 1.51518
0.7 1.514995
0.72 1.515485
0.74 1.51566
0.76 1.516025
0.78 1.51531
};
\addlegendentry{\small DL code with $1$-bit CRC}
\addplot [semithick, green01270, mark=square*, mark size=2, mark options={solid}]
table {%
0.4 1.743455
0.42 1.682395
0.44 1.624195
0.46 1.567565
0.48 1.52988
0.5 1.530075
0.52 1.529315
0.54 1.52919
0.56 1.529045
0.58 1.52955
0.6 1.527985
0.62 1.529545
0.64 1.528725
0.66 1.52853
0.68 1.52891
0.7 1.529355
0.72 1.52815
0.74 1.52833
0.76 1.52847
0.78 1.529625
};
\addlegendentry{\small DL code (genie)}

\addplot [semithick, red, mark=*, mark size=2, mark options={solid}]
table {%
0.4 1.730455
0.42 1.671675
0.44 1.616675
0.46 1.565095
0.48 1.555535
0.5 1.55555
0.52 1.55438
0.54 1.55508
0.56 1.55529
0.58 1.5548
0.6 1.554865
0.62 1.553695
0.64 1.555035
0.66 1.55423
0.68 1.5543
0.7 1.55486
0.72 1.555085
0.74 1.55484
0.76 1.55467
0.78 1.55441
};
\addlegendentry{\small Cyclic code with $1$-bit CRC}
\addplot [semithick, blue, mark=asterisk, mark size=2, mark options={solid}]
table {%
0.4 1.747215
0.42 1.687085
0.44 1.63033
0.46 1.607745
0.48 1.60601
0.5 1.60715
0.52 1.604865
0.54 1.609325
0.56 1.607545
0.58 1.60454
0.6 1.606835
0.62 1.607555
0.64 1.605165
0.66 1.608065
0.68 1.605755
0.7 1.607235
0.72 1.606205
0.74 1.605225
0.76 1.605595
0.78 1.60797
};
\addlegendentry{\small Cyclic code (genie)}

\addplot [semithick, goldenrod1911910, mark=triangle*, mark size=2, mark options={solid}]
table {%
0.4 1.722235
0.42 1.659665
0.44 1.60166
0.46 1.54878
0.48 1.51053
0.5 1.506635
0.52 1.50662
0.54 1.506895
0.56 1.50641
0.58 1.50536
0.6 1.506025
0.62 1.50633
0.64 1.50647
0.66 1.505925
0.68 1.506575
0.7 1.50671
0.72 1.506275
0.74 1.505965
0.76 1.50582
0.78 1.5065
};
\addlegendentry{\small Using PPV expression (genie)}

\end{axis}

\end{tikzpicture}

%% file: aoi_plots_cyclic_ae_crc/test.tex
\begin{tikzpicture}[scale=0.9]
\definecolor{brown}{RGB}{165,42,42}
\definecolor{dimgray85}{RGB}{85,85,85}
\definecolor{gainsboro229}{RGB}{229,229,229}
\definecolor{gray}{RGB}{128,128,128}
\definecolor{green}{RGB}{0,128,0}
\definecolor{orange}{RGB}{255,165,0}
\definecolor{pink}{RGB}{255,192,203}
\definecolor{purple}{RGB}{128,0,128}
\begin{axis}[xlabel=,
    ylabel=$A_{\rm SRP}^{\mathsf{D}, \mathsf{C}}$,xmin=0.7, xmax=3.4, 
    xtick={0.8,3.2},
    xticklabels={$1$-bit CRC, $3$-bit CRC},
    legend style={at={(0.9,1)},anchor=north east}]
    \addplot[
        scatter,only marks,scatter src=explicit symbolic,
        scatter/classes={
            a={mark=square*,blue},
            b={mark=triangle*,red},
            c={mark=o,draw=black,fill=black},
            d={mark=triangle*,black},
            e={mark=square*,brown},
            f={mark=triangle*,orange},
            g={mark=o,draw=purple,fill=purple},
            h={mark=square*,black}
        }
    ]
    table[x=x,y=y,meta=label]{
        x    y    label
        0.8   2.08160952 a
        3.2   2.10555685 b
        0.8   2.04269493 c
        3.2   2.09306834 d
         0.8    2.02293254 e
        3.2   2.07714795 f
         0.8    2.01185243 g
        3.2   2.06788077 h
        
    };
    \legend{\footnotesize Cyclic code ($1$-bit CRC + genie),\footnotesize Cyclic code ($3$-bit CRC + genie),\footnotesize Cyclic code ($1$-bit CRC),\footnotesize Cyclic code ($3$-bit CRC), \footnotesize DL code ($1$-bit CRC + genie),\footnotesize DL code ($3$-bit CRC + genie),\footnotesize DL code ($1$-bit CRC), \footnotesize DL code ($3$-bit CRC)}
\end{axis}
\end{tikzpicture}

%% file: aoi_plots_cyclic_ae_crc/crc1_vs_crc3_dpp.tex
\begin{tikzpicture}[scale=0.9]
\definecolor{brown}{RGB}{165,42,42}
\definecolor{dimgray85}{RGB}{85,85,85}
\definecolor{gainsboro229}{RGB}{229,229,229}
\definecolor{gray}{RGB}{128,128,128}
\definecolor{green}{RGB}{0,128,0}
\definecolor{orange}{RGB}{255,165,0}
\definecolor{pink}{RGB}{255,192,203}
\definecolor{purple}{RGB}{128,0,128}
\begin{axis}[xlabel=,
    ylabel=$A_{\rm DPP}^{\mathsf{D}, \mathsf{C}}$,xmin=0.7, xmax=3.4, 
    xtick={0.8,3.2},
    xticklabels={$1$-bit CRC, $3$-bit CRC},
    legend style={at={(0.9,1)},anchor=north east}]
    \addplot[
        scatter,only marks,scatter src=explicit symbolic,
        scatter/classes={
            a={mark=square*,blue},
            b={mark=triangle*,red},
            c={mark=o,draw=black,fill=black},
            d={mark=triangle*,black},
            e={mark=square*,brown},
            f={mark=triangle*,orange},
            g={mark=o,draw=purple,fill=purple},
            h={mark=square*,black}
        }
    ]
    table[x=x,y=y,meta=label]{
        x    y    label
        0.8   1.607 a
        3.2   1.6389 b
        0.8   1.5325 c
        3.2   1.6218 d
         0.8    1.5166 e
        3.2   1.5952 f
         0.8    1.5079 g
        3.2   1.585 h
        
    };
    \legend{\footnotesize Cyclic code ($1$-bit CRC + genie),\footnotesize Cyclic code ($3$-bit CRC + genie),\footnotesize Cyclic code ($1$-bit CRC),\footnotesize Cyclic code ($3$-bit CRC),\footnotesize DL code ($1$-bit CRC + genie),\footnotesize DL code ($3$-bit CRC + genie),\footnotesize DL code ($1$-bit CRC),\footnotesize DL code ($3$-bit CRC)}
\end{axis}
\end{tikzpicture}

%% file: aoi_plots_cyclic_ae_crc/AOI_POWER_DISTORTION_AE_CRC.tex
\begin{tikzpicture}

\definecolor{darkcyan32144140}{RGB}{32,144,140}
\definecolor{darkslateblue5882139}{RGB}{58,82,139}
\definecolor{darkslateblue7143122}{RGB}{71,43,122}
\definecolor{dimgray85}{RGB}{85,85,85}
\definecolor{gainsboro229}{RGB}{229,229,229}
\definecolor{gold25323136}{RGB}{253,231,36}
\definecolor{indigo68184}{RGB}{68,1,84}
\definecolor{mediumseagreen40174127}{RGB}{40,174,127}
\definecolor{mediumseagreen9120098}{RGB}{91,200,98}
\definecolor{teal44113142}{RGB}{44,113,142}
\definecolor{yellowgreen17322048}{RGB}{173,220,48}

\begin{axis}[
axis background/.style={},
axis line style={black},
tick align=outside,
tick pos=left,
title={AoI},
x grid style={gray!20},
 grid=both,
grid style={line width=.1pt, draw=gray!10},
major grid style={line width=.2pt,draw=gray!50},
xlabel=\textcolor{dimgray85}{Bound on average power, $\bar{P}$},
xmajorgrids,
xmin=0.25, xmax=0.75,
xtick style={color=dimgray85},
y grid style={gray!20},
ylabel=\textcolor{dimgray85}{Bound on average distortion, $\bar{d}_1=\bar{d}_2=\bar{d}$},
ymajorgrids,
ymin=0.4, ymax=0.59,
ytick style={color=dimgray85}
]
\addplot [draw=darkslateblue7143122]
table{%
x  y
0.745 0.463483338127065
0.74 0.463616621697279
0.735 0.463748691552267
0.73 0.463885502148125
0.725000000000001 0.464024665113354
0.72 0.464163212941242
0.715 0.464301148545641
0.71 0.464438490709648
0.705 0.464575197386942
0.7 0.46471522449349
0.695 0.46486043434223
0.69023949347141 0.464998643452176
};
\addplot [draw=darkslateblue7143122]
table{%
x  y
0.64976039926657 0.46617162658588
0.645 0.466309301422207
0.64 0.46645384524175
0.635 0.466598328586256
0.63 0.46674276993662
0.625 0.466887108923522
0.62 0.4670314115814
0.615 0.467175659327083
0.61 0.467319849360721
0.605 0.467463966926024
0.6 0.467608019354423
0.595 0.467752006839945
0.59 0.467895941728547
0.585 0.468039820198065
0.58 0.468183642581791
0.575 0.468327376558476
0.57 0.468471067788576
0.565 0.468614692757188
0.56 0.468758260518435
0.555 0.468901771325891
0.55 0.469045208500251
0.545 0.46918858960898
0.54 0.469331931494527
0.535 0.469475170533946
0.53 0.469618369889799
0.525 0.469761511054227
0.52 0.469904580220665
0.516664461646432 0.47
0.515 0.470058464248533
0.51 0.470233531186582
0.505 0.470407650738662
0.5 0.470580891816866
0.495 0.47081797291858
0.49 0.471059552173314
0.485 0.471322232550962
0.48 0.471655441634989
0.475 0.472031681660548
0.47 0.472459968870155
0.465 0.472952123139961
0.46 0.473523687284726
0.455 0.474195816825991
0.45 0.475006336379665
0.445 0.476054216683743
0.44 0.477300473767894
0.435 0.478807781068777
0.431671866502925 0.48
0.431671558788549 0.49
0.431671678067838 0.5
0.431671735450281 0.51
0.431671775728718 0.52
0.431671810429811 0.53
0.431671831963639 0.54
0.431671661549291 0.55
0.431671824274848 0.56
0.431671883460417 0.57
0.431671860445329 0.58
0.431671767937778 0.59
};
\addplot [draw=darkslateblue5882139]
table{%
x  y
0.745 0.449826396840645
0.74 0.449909979954921
0.735 0.449993500055446
0.734610945729397 0.45
0.73 0.45008612178957
0.725 0.450180189482581
0.72 0.450274972616926
0.715 0.450370479537347
0.71 0.450466729869443
0.705 0.450563698379616
0.7 0.450666359203535
0.695 0.450785177316979
0.69 0.450904406616282
0.685 0.451024041069717
0.68 0.45114408944472
0.675 0.45126455078055
0.67 0.451385428103942
0.665 0.451506723874858
0.66 0.451628440386628
0.655 0.451750580492194
0.65 0.451873146607357
0.645 0.451996141038704
0.64 0.452119566580951
0.635 0.452243426040481
0.63 0.452367721634868
0.625 0.452492456182569
0.62 0.452617633502042
0.615 0.45274325456862
0.61 0.45286932037258
0.605 0.452995837783746
0.6 0.453138592389316
0.595 0.453281507060148
0.59 0.453424235112538
0.585 0.453566834455806
0.58 0.453709281537791
0.575 0.453851561881068
0.57 0.453993705720928
0.565 0.454135700697495
0.56 0.454277546561133
0.555 0.45441924101118
0.55 0.454560791966436
0.545303504376156 0.4546936012085
};
\addplot [draw=darkslateblue5882139]
table{%
x  y
0.504696236401222 0.455836536004489
0.5 0.455968386324594
0.495 0.456149825228942
0.49 0.456331625292984
0.485 0.456513370340389
0.48 0.45669514417912
0.475 0.456876921060738
0.47 0.457058698194145
0.465 0.457240476144208
0.46 0.45742225125295
0.455 0.457604035335056
0.45 0.457785821247135
0.445 0.457967601945712
0.44 0.458149375761854
0.435 0.458331164537914
0.43 0.458512945148597
0.425 0.458694747581449
0.42 0.458876519474714
0.415 0.459058308510215
0.41 0.459240098188362
0.405 0.459421888984042
0.4 0.459603703952821
0.395 0.459785474642997
0.39 0.45996730063331
0.389100086999783 0.46
0.385 0.460182404950838
0.38 0.460427335310952
0.375 0.460701212882411
0.37 0.461009577512685
0.365 0.461359525974421
0.36 0.461760149112469
0.355 0.46222354623369
0.35 0.462765937842828
0.345 0.463409549505294
0.34 0.464185987957783
0.335 0.465141465977428
0.33 0.466346416572684
0.325 0.467913720948753
0.32007407622793 0.47
0.320074096854453 0.48
0.320074301440282 0.49
0.320074296216315 0.5
0.320074294293551 0.51
0.320074307114979 0.52
0.320074097447631 0.53
0.320074052618746 0.54
0.320074186013973 0.55
0.320074306991673 0.56
0.320074468849919 0.57
0.320074337327841 0.58
0.320074364481163 0.59
};
\addplot [ draw=teal44113142]
table{%
x  y
0.745000000000001 0.438013184992424
0.74 0.438097113312988
0.735 0.438181033017481
0.73 0.438265108661102
0.725 0.43834917884369
0.72 0.438433221697652
0.715 0.438517231705002
0.71 0.438601219996829
0.705 0.438685173135386
0.7 0.438769095487145
0.695 0.438852988920772
0.69 0.438936855133096
0.685 0.439020692357179
0.68 0.439104503091814
0.675 0.439188279400414
0.67 0.439272029039261
0.665 0.439355751601766
0.66 0.439439441383416
0.655 0.439523103058911
0.65 0.43960673588622
0.645 0.439690339624178
0.64 0.439773914200026
0.635 0.43985745959742
0.63 0.439940977021693
0.626465446480155 0.44
0.625 0.440026235697078
0.62 0.440135065711765
0.615 0.44024513897771
0.61 0.44035516649528
0.605 0.440465148798923
0.6 0.440576234250499
0.595 0.440687745437669
0.59 0.440799663240818
0.585 0.440911995546077
0.58 0.441024742712346
0.575 0.441137904827455
0.57 0.441251488923607
0.565 0.441365495748563
0.56 0.44147992803525
0.555 0.441594788447649
0.55 0.441710079867805
0.545 0.441825804946315
0.54 0.441941966384656
0.535 0.44205857128956
0.53 0.442175612381991
0.525 0.44229309862205
0.52 0.442411035797224
0.515 0.442529421680095
0.51 0.442648258674043
0.505 0.442767553361844
0.5 0.44288730860744
0.495 0.443074389930342
0.49 0.443260568071727
0.485 0.443445842617702
0.48 0.443630223981683
0.475 0.443813717310673
0.47 0.443996327825863
0.465 0.444178061321921
0.46 0.444358924695615
0.455 0.444538914745172
0.45 0.444718044239446
0.445 0.44489631796302
0.44 0.445073741082431
0.435 0.445250318996862
0.43 0.44542605581589
0.425 0.445600963044106
0.42 0.445775029854037
0.415 0.44594828112063
0.41 0.446120711072904
0.405 0.446292314146683
0.4 0.446463114479104
0.395284317859472 0.446623446027266
};
\addplot [ draw=teal44113142]
table{%
x  y
0.354716660989528 0.448092817390751
0.35 0.448264075818043
0.345 0.448445435683934
0.34 0.448626599482743
0.335 0.448807566433613
0.33 0.448988338699077
0.325 0.449168922281324
0.32 0.44934931658369
0.315 0.449529519985611
0.31 0.449709552178468
0.305 0.449889365377851
0.301921599908943 0.45
0.3 0.450077876733909
0.295 0.45028275112552
0.29 0.450490928168401
0.285 0.450721141316522
0.28 0.45097683649256
0.275 0.451256743318915
0.27 0.451564552055952
0.265 0.451904743096947
0.26 0.452282812017052
0.255 0.452705586512826
0.25 0.453181566726614
};
\addplot [draw=darkcyan32144140]
table{%
x  y
0.745 0.426779948429913
0.74 0.426858669539933
0.735 0.42693753920354
0.73 0.427017170883003
0.725000000000001 0.427097068843852
0.72 0.427177126459114
0.715 0.427257329436489
0.71 0.427337722676838
0.705 0.427418262846125
0.7 0.427500439881964
0.695000000000001 0.427583858449226
0.690260928343891 0.427662929890837
};
\addplot [draw=darkcyan32144140]
table{%
x  y
0.649739059852205 0.42833862609819
0.645 0.428417599860074
0.64 0.428500910886579
0.635 0.428584210512792
0.63 0.428667498846468
0.625 0.428750775388311
0.62 0.428834040655147
0.615 0.428917294330243
0.61 0.429000538683866
0.605 0.429083767288315
0.6 0.429166986577692
0.595 0.429250194242393
0.59 0.42933339062959
0.585 0.429416575162691
0.58 0.429499757937628
0.575 0.429582910013144
0.57 0.42966606015367
0.565 0.429749198878107
0.56 0.42983234271155
0.555 0.429915441403233
0.55 0.429998547485508
0.549923771252969 0.43
0.545 0.430101101237934
0.54 0.430211218676428
0.535 0.430321307138584
0.53 0.430431348484528
0.525 0.430541343616081
0.52 0.43065129185271
0.515 0.430761193728115
0.51 0.430871050174757
0.505 0.430980857592618
0.5 0.431090620020977
0.495 0.431277071716994
0.49 0.43146335584249
0.485 0.431649523091885
0.48 0.431835553981654
0.475 0.432021452664522
0.47 0.432207213377259
0.465 0.432392841259033
0.46 0.432578337787526
0.455 0.432763696403893
0.45 0.432949249870867
0.445 0.433135481026383
0.44 0.433321519126395
0.435 0.433507360515066
0.43 0.433693019450239
0.425 0.433878457316586
0.42 0.434063704422482
0.415 0.434248760438624
0.41 0.434433623216912
0.405 0.434618277936014
0.4 0.434802742826996
0.395 0.434987014380892
0.39 0.435171090870816
0.385 0.435354961258445
0.38 0.435538642396325
0.375 0.435722145839188
0.37 0.435906208434997
0.365 0.436090309381427
0.36 0.436274270865044
0.355 0.43645809444275
0.35 0.436641783933837
0.345 0.43682534767903
0.34 0.437008760239434
0.335 0.437192033006648
0.33 0.437375175688311
0.325 0.437558182017987
0.32 0.437741051872221
0.315 0.437923785603308
0.31 0.43810638432353
0.305 0.438288852618252
0.3 0.438471174692314
0.295 0.438653351900979
0.29 0.438835402579859
0.285 0.439017313353117
0.28 0.439199089858879
0.275 0.439380730153453
0.27 0.439562261154859
0.265 0.439743619248372
0.26 0.439924844313024
0.257925699982123 0.44
0.255 0.44011025041785
0.25 0.440298599202947
};
\addplot [draw=mediumseagreen40174127]
table{%
x  y
0.745 0.416166473942116
0.74 0.416238278880858
0.735 0.416310051861035
0.73 0.416381789289757
0.725 0.416453516870104
0.72 0.416525221235983
0.715 0.416596903762803
0.71 0.416668564267449
0.705 0.416740209087074
0.7 0.416815668091645
0.695 0.416894386558374
0.69 0.4169732445341
0.685 0.417052270180704
0.68 0.417131400365418
0.675 0.417210686798775
0.67 0.417290117504847
0.665 0.417369691504582
0.66 0.417449409868804
0.655 0.41753054117502
0.65 0.417613092189625
0.645 0.417695619596651
0.64 0.417778107243397
0.635 0.417860570347099
0.63 0.417943003891823
0.625 0.418025407846489
0.62 0.41810778225537
0.615 0.418190127177537
0.61 0.418272459372891
0.605 0.418354728408472
0.6 0.418436984772586
0.595 0.418519211576851
0.59 0.418601411073435
0.585 0.418683576690367
0.58 0.418765714943881
0.575 0.418847823736463
0.57 0.418929904027231
0.565 0.419011952280412
0.56 0.419093972585158
0.555 0.419175964956432
0.55 0.419257929397564
0.545 0.419339872132711
0.54 0.419421758950228
0.535 0.419503629278066
0.53 0.419585472171658
0.525 0.419667285614665
0.52 0.419749069455048
0.515 0.41983082520017
0.51 0.419912550621722
0.505 0.419994244319939
0.504647722494633 0.42
0.5 0.420081086456579
0.495 0.420246134882081
0.49 0.420411719832348
0.485 0.420577825105106
0.48 0.420744473326284
0.475 0.420918255312067
0.47 0.421101322444403
0.465 0.421284379928645
0.46 0.421467365786385
0.455 0.421650339632006
0.45 0.421833367144453
0.445 0.42201719874962
0.44 0.422204228501193
0.435 0.422391102181133
0.43 0.422577838348099
0.425 0.422764417291406
0.42 0.422950850809084
0.415 0.42313713755259
0.41 0.423323292616512
0.405 0.42350927419326
0.4 0.423695112615935
0.395219266237629 0.423872665402829
};
\addplot [draw=mediumseagreen40174127]
table{%
x  y
0.354780395355421 0.425369159934033
0.35 0.425545429877995
0.345 0.425729668094323
0.34 0.425913734915584
0.335 0.426097645907467
0.33 0.426281422265738
0.325 0.426465055072568
0.32 0.426648541994674
0.315 0.426831869678498
0.31 0.427015059095801
0.305 0.427198098859008
0.3 0.427380989428194
0.295 0.427563735844884
0.29 0.427746339319307
0.285 0.427928778734628
0.28 0.428111074400024
0.275 0.428293228387139
0.27 0.428475265998409
0.265 0.428657094795559
0.26 0.428838838963701
0.255 0.429020431362063
0.25 0.429202041532042
};
\addplot [draw=mediumseagreen9120098]
table{%
x  y
0.745 0.406179413138123
0.74 0.40622154746226
0.735 0.406264241977566
0.73 0.406307554255998
0.725 0.406351484472568
0.72 0.406396048634557
0.715 0.406446491966986
0.71 0.406518189600997
0.705 0.40658986524714
0.7 0.406661516151607
0.695 0.406733160714132
0.690263729528741 0.406800975908363
};
\addplot [draw=mediumseagreen9120098]
table{%
x  y
0.649736903141727 0.407392155498224
0.645 0.407467197784622
0.64 0.407546517369371
0.635 0.407625993620248
0.63 0.407705611685676
0.625 0.407785371636176
0.62 0.407865273751049
0.615 0.407945319174434
0.61 0.408026374190438
0.605 0.408108727531434
0.6 0.408191044913309
0.595 0.408273335116685
0.59 0.408355595604712
0.585 0.4084378251257
0.58 0.408520025097356
0.575 0.408602195293481
0.57 0.408684343447661
0.565 0.408766441888037
0.56 0.408848520725033
0.555 0.408930569373574
0.55 0.409012587834521
0.545 0.409094578887278
0.54 0.409176534088889
0.535 0.409258461936747
0.53 0.409340359538806
0.525 0.409422230929261
0.52 0.409504064174048
0.515 0.409585871246325
0.51 0.40966764916002
0.505 0.409749394782428
0.5 0.40983111755479
0.495 0.409985270056235
0.494522213470751 0.41
0.49 0.410145498467307
0.485 0.410306271301075
0.48 0.410466937038351
0.475 0.410627917748362
0.47 0.410789746447884
0.465 0.410951991499465
0.46 0.411114648650669
0.455 0.411277725907389
0.45 0.41144122381875
0.445 0.41160525425827
0.44 0.411770412314885
0.435 0.411936149826308
0.43 0.412102468709243
0.425 0.412269377073597
0.42 0.41243687862001
0.415 0.412605603228741
0.41 0.4127866963332
0.405 0.412967810887376
0.4 0.413148959787526
0.395 0.413330140449628
0.39 0.413511352522529
0.385 0.413692596769132
0.38 0.413873872732722
0.375 0.414055186977506
0.37 0.414236520070119
0.365 0.414417890825778
0.36 0.41459929352037
0.355 0.414780728323148
0.35 0.414962195091273
0.345 0.415143693409429
0.34 0.415325223967023
0.335 0.415506782995546
0.33 0.415688381481232
0.325 0.415870013205373
0.32 0.416051682506794
0.315 0.416233360196975
0.31 0.416415083295245
0.305 0.41659683596046
0.3 0.416778624804513
0.295 0.416961706636838
0.29 0.417144863511997
0.285 0.417327893913683
0.28 0.417510743675115
0.275 0.417693469912291
0.27 0.417876054028051
0.265 0.418058497862043
0.26 0.41824079694481
0.255 0.418422946007095
0.25 0.418604957516453
};
\addplot [ draw=yellowgreen17322048]
table{%
x  y
0.500807752020938 0.4
0.5 0.400013777254917
0.495 0.400174713466843
0.49 0.400335535372421
0.485 0.400496255781315
0.48 0.40065686937127
0.475 0.400817351810659
0.47 0.400977734682584
0.465 0.401138008477935
0.46 0.401298172873451
0.455 0.401458227694743
0.450232702231655 0.401610728713699
};
\addplot [ draw=yellowgreen17322048]
table{%
x  y
0.409767702960488 0.402906290574347
0.405 0.403062816108843
0.4 0.403227396420644
0.395 0.403392386569392
0.39 0.403557803925656
0.385 0.403723643361401
0.38 0.403889915595375
0.375 0.404056617880276
0.37 0.404223752283046
0.365 0.404391321140573
0.36 0.404559327195051
0.355 0.404727771684096
0.35 0.404902612518176
0.345 0.40508019308359
0.34 0.405257637883564
0.335 0.405434945506047
0.33 0.405612118429653
0.325 0.405789151899564
0.32 0.405966051805337
0.315 0.406142815140517
0.31 0.406319439237836
0.305 0.406495921606264
0.3 0.406672276965296
0.295 0.406851178896097
0.29 0.407030532210281
0.285 0.407209909927001
0.28 0.407389308926112
0.275 0.407568739905915
0.27 0.407748198034063
0.265 0.407927683115539
0.26 0.408107196249411
0.255 0.408286734257786
0.25 0.408466299885843
};
\draw (axis cs:0.67,0.465585626092377) node[
  scale=0.5,
  text=darkslateblue7143122,
  rotate=356.8
]{2.05};
\draw (axis cs:0.525,0.455266290466709) node[
  scale=0.5,
  text=darkslateblue5882139,
  rotate=356.9
]{2.10};
\draw (axis cs:0.375,0.447354483610956) node[
  scale=0.5,
  text=teal44113142,
  rotate=356.0
]{2.15};
\draw (axis cs:0.67,0.428000871909085) node[
  scale=0.5,
  text=darkcyan32144140,
  rotate=358.1
]{2.20};
\draw (axis cs:0.375,0.424622125863071) node[
  scale=0.5,
  text=mediumseagreen40174127,
  rotate=355.9
]{2.25};
\draw (axis cs:0.67,0.407090949424569) node[
  scale=0.5,
  text=mediumseagreen9120098,
  rotate=358.4
]{2.30};
\draw (axis cs:0.43,0.402256853871309) node[
  scale=0.5,
  text=yellowgreen17322048,
  rotate=356.5
]{2.35};
\end{axis}

\end{tikzpicture}

%% file: aoi_plots_cyclic_ae_crc/AOI_POWER_DISTORTION_DPP_AE_CRC1_updated.tex
\begin{tikzpicture}

\definecolor{darkcyan32144140}{RGB}{32,144,140}
\definecolor{darkgray176}{RGB}{176,176,176}
\definecolor{darkslateblue48103141}{RGB}{48,103,141}
\definecolor{darkslateblue5882139}{RGB}{58,82,139}
\definecolor{dimgray85}{RGB}{85,85,85}
\definecolor{darkslateblue6857130}{RGB}{68,57,130}
\definecolor{gold25323136}{RGB}{253,231,36}
\definecolor{greenyellow19922431}{RGB}{199,224,31}
\definecolor{indigo68184}{RGB}{68,1,84}
\definecolor{indigo7230112}{RGB}{72,30,112}
\definecolor{mediumseagreen32164133}{RGB}{32,164,133}
\definecolor{mediumseagreen53183120}{RGB}{53,183,120}
\definecolor{mediumseagreen9120098}{RGB}{91,200,98}
\definecolor{teal40123142}{RGB}{40,123,142}
\definecolor{yellowgreen14421467}{RGB}{144,214,67}

\begin{axis}[
axis background/.style={},
axis line style={black},
tick align=outside,
tick pos=left,
title={AoI},
x grid style={gray!20},
 grid=both,
 grid style={line width=.1pt, draw=gray!10},
major grid style={line width=.2pt,draw=gray!50},
xlabel=\textcolor{dimgray85}{Bound on average power, $\bar{P}$},
xmin=0.05, xmax=0.5,
xtick style={color=dimgray85},
y grid style={gray!20},
ylabel=\textcolor{dimgray85}{Bound on average distortion, $\bar{d}_1=\bar{d}_2=\bar{d}$},
ymin=0.4, ymax=0.8,
ytick style={color=dimgray85}
]
\addplot [draw=indigo7230112]
table{%
x  y
0.5 0.444860085301155
0.45 0.447726075504829
0.4 0.45135847818957
0.35 0.456513271556463
0.3 0.462812949640288
0.25 0.47948748368727
0.249907338557456 0.48
0.249602153584585 0.52
0.249913568483129 0.542895482853349
};
\addplot [draw=indigo7230112]
table{%
x  y
0.251838314057578 0.577381576287451
0.251544846668204 0.6
0.25 0.621612903225802
0.249902796725784 0.64
0.25 0.661923076923081
0.25112763915547 0.68
0.252068437180797 0.72
0.25 0.740571428571426
0.249871468175072 0.76
0.24993147751606 0.8
};
\addplot [draw=darkslateblue6857130]
table{%
x  y
0.199589296603178 0.4
0.199451212229786 0.44
0.199558122915674 0.48
0.199288625151927 0.52
0.199540639797834 0.56
0.199375586765238 0.582617000646154
};
\addplot [draw=darkslateblue6857130]
table{%
x  y
0.199389449746289 0.617382847527972
0.199572542715961 0.64
0.199474112417314 0.68
0.199365467252208 0.72
0.199356484019353 0.76
0.199103299856528 0.8
};
\addplot [draw=darkslateblue5882139]
table{%
x  y
0.163720004304315 0.4
0.163621760420399 0.44
0.163826822296332 0.48
0.163541145349253 0.52
0.164114403126956 0.56
0.163894436524726 0.582752770246714
};
\addplot [draw=darkslateblue5882139]
table{%
x  y
0.163781244072893 0.617248258343901
0.163851878311603 0.64
0.164100508200781 0.68
0.163949685980073 0.72
0.163605596205639 0.76
0.163823294211727 0.8
};
\addplot [draw=darkslateblue48103141]
table{%
x  y
0.138958952933253 0.4
0.138989779967315 0.44
0.139001058568122 0.48
0.138942290576382 0.52
0.139266002117359 0.56
0.139058087568658 0.582752648036605
};
\addplot [draw=darkslateblue48103141]
table{%
x  y
0.138887919590287 0.617248370115088
0.138871350840412 0.64
0.139272633218775 0.68
0.13908133977038 0.72
0.139025182490373 0.76
0.139004959534298 0.8
};
\addplot [ draw=teal40123142]
table{%
x  y
0.121078714768481 0.4
0.121226107861017 0.44
0.121059225989343 0.48
0.121141728078607 0.52
0.121423091105904 0.56
0.12107498717439 0.582686882507097
};
\addplot [ draw=teal40123142]
table{%
x  y
0.120759286475589 0.617315914385868
0.120693719022286 0.64
0.121434814665065 0.68
0.121067216697891 0.72
0.121307670501524 0.76
0.120926462701048 0.8
};
\addplot [draw=darkcyan32144140]
table{%
x  y
0.103198476603708 0.4
0.103462435754719 0.44
0.103117393410563 0.48
0.103341165580832 0.52
0.103580180094448 0.56
0.103091225072079 0.582689701123047
};
\addplot [draw=darkcyan32144140]
table{%
x  y
0.102630703618713 0.617315702799436
0.102516087204159 0.64
0.103596996111356 0.68
0.103053093625402 0.72
0.103590158512674 0.76
0.102847965867798 0.8
};
\addplot [draw=mediumseagreen32164133]
table{%
x  y
0.0949388037714218 0.4
0.0950251394979516 0.44
0.0949063361787298 0.48
0.0949831288690887 0.52
0.0950809331673166 0.56
0.0949020480570266 0.582752381916842
};
\addplot [draw=mediumseagreen32164133]
table{%
x  y
0.0947297001621264 0.61724832093583
0.0946812341553353 0.64
0.0950430695487321 0.68
0.0948632495196312 0.72
0.0950732088093605 0.76
0.094815706722515 0.8
};
\addplot [draw=mediumseagreen53183120]
table{%
x  y
0.0887750070370013 0.4
0.0888458282439839 0.44
0.0887415388156191 0.48
0.0888069979680529 0.52
0.0889270996367187 0.56
0.0887453111175842 0.582752406736358
};
\addplot [draw=mediumseagreen53183120]
table{%
x  y
0.0885645927593103 0.617248300663335
0.0885079881802917 0.64
0.0888341004859553 0.68
0.0886782436232589 0.72
0.0888943811801149 0.76
0.0886619994215515 0.8
};
\addplot [draw=mediumseagreen9120098]
table{%
x  y
0.0826112103025808 0.4
0.0826665169900163 0.44
0.0825767414525085 0.48
0.0826308670670172 0.52
0.0827732661061208 0.56
0.0825885741685651 0.582752431955331
};
\addplot [draw=mediumseagreen9120098]
table{%
x  y
0.0823994853799342 0.617248277250407
0.0823347422052481 0.64
0.0826251314231785 0.68
0.0824932377268866 0.72
0.0827155535508693 0.76
0.0825082921205881 0.8
};
\addplot [ draw=yellowgreen14421467]
table{%
x  y
0.0764474135681603 0.4
0.0764872057360486 0.44
0.0764119440893978 0.48
0.0764547361659815 0.52
0.0766194325755229 0.56
0.0764318372098167 0.582752457573753
};
\addplot [draw=none, draw=yellowgreen14421467]
table{%
x  y
0.0762343780273679 0.617248250697084
0.0761614962302044 0.64
0.0764161623604017 0.68
0.0763082318305143 0.72
0.0765367259216236 0.76
0.0763545848196246 0.8
};
\addplot [draw=greenyellow19922431]
table{%
x  y
0.0702836168337398 0.4
0.070307894482081 0.44
0.0702471467262871 0.48
0.0702786052649457 0.52
0.070465599044925 0.56
0.070275666545447 0.582684846272372
};
\addplot [draw=greenyellow19922431]
table{%
x  y
0.0700690298345906 0.617315861085777
0.0699882502551608 0.64
0.0702071932976248 0.68
0.0701232259341421 0.72
0.070357898292378 0.76
0.0702008775186611 0.8
};
\addplot [draw=gold25323136]
table{%
x  y
0.0641198200993193 0.4
0.0641285832281134 0.44
0.0640823493631764 0.48
0.06410247436391 0.52
0.0643117655143271 0.56
0.0641189381964423 0.582684872793276
};
\addplot [draw=gold25323136]
table{%
x  y
0.063903898350669 0.617315828123041
0.0638150042801172 0.64
0.063998224234848 0.68
0.0639382200377698 0.72
0.0641790706631324 0.76
0.0640471702176976 0.8
};
\draw (axis cs:0.252063834893208,0.56) node[
  scale=0.5,
  text=indigo7230112,
  rotate=84.5
]{1.6};
\draw (axis cs:0.199248730123021,0.6) node[
  scale=0.5,
  text=darkslateblue6857130,
  rotate=90.0
]{1.9};
\draw (axis cs:0.1637276957505,0.6) node[
  scale=0.5,
  text=darkslateblue5882139,
  rotate=270.2
]{2.2};
\draw (axis cs:0.13890048062137,0.6) node[
  scale=0.5,
  text=darkslateblue48103141,
  rotate=270.3
]{2.5};
\draw (axis cs:0.120809337442033,0.6) node[
  scale=0.5,
  text=teal40123142,
  rotate=270.6
]{2.8};
\draw (axis cs:0.102718194262695,0.6) node[
  scale=0.5,
  text=darkcyan32144140,
  rotate=270.9
]{3.1};
\draw (axis cs:0.0947664428235034,0.6) node[
  scale=0.5,
  text=mediumseagreen32164133,
  rotate=270.3
]{3.4};
\draw (axis cs:0.0886075052812889,0.6) node[
  scale=0.5,
  text=mediumseagreen53183120,
  rotate=270.4
]{3.7};
\draw (axis cs:0.0824485677390746,0.6) node[
  scale=0.5,
  text=mediumseagreen9120098,
  rotate=270.4
]{4.0};
\draw (axis cs:0.0762896301968602,0.6) node[
  scale=0.5,
  text=yellowgreen14421467,
  rotate=270.4
]{4.3};
\draw (axis cs:0.0701306926546458,0.6) node[
  scale=0.5,
  text=greenyellow19922431,
  rotate=270.4
]{4.6};
\draw (axis cs:0.0639717551124314,0.6) node[
  scale=0.5,
  text=gold25323136,
  rotate=270.4
]{4.9};
\end{axis}

\end{tikzpicture}